\PassOptionsToPackage{table}{xcolor} %
\documentclass[sigconf]{acmart} %

\usepackage[table]{xcolor}

\usepackage{hyphenat}
\usepackage{url}
\usepackage{xspace} %
\usepackage{amsmath}
\usepackage{mathtools}
\usepackage{amsfonts}
\usepackage{algorithm}
\usepackage[noend]{algpseudocode}
\usepackage{chngcntr} %
\counterwithout{algorithm}{section}
\usepackage{balance}
\usepackage{enumitem}
\usepackage{soul}
\usepackage{tikz}
\usepackage{tabularx}
\usetikzlibrary{fit,positioning,calc,shapes}
\usepackage{pgfplots}
\pgfplotsset{compat=1.13}
\usepgfplotslibrary{colorbrewer}
\usepackage{multirow}
\usepackage[caption=false]{subfig}
\usepackage{breqn}
\usepackage{caption}
\usepackage{pifont}
\usepackage{stfloats}
\usepackage{booktabs}

\copyrightyear{2018}
\acmYear{2018} 
\setcopyright{iw3c2w3}
\acmConference[WWW 2018]{The 2018 Web Conference}{April 23--27, 2018}{Lyon, France}
\acmBooktitle{WWW 2018: The 2018 Web Conference, April 23--27, 2018, Lyon, France}
\acmPrice{}
\acmDOI{10.1145/3178876.3186120}
\acmISBN{978-1-4503-5639-8/18/04}

\definecolor{cycle1}{RGB}{228, 26, 28}
\definecolor{cycle2}{RGB}{55, 126, 184}
\definecolor{cycle3}{RGB}{77, 175, 74}
\definecolor{cycle4}{RGB}{152, 78, 163}
\definecolor{cycle5}{RGB}{255, 127, 0}
\definecolor{cycle6}{RGB}{153, 153, 153}%
\definecolor{cycle7}{RGB}{166, 86, 40}
\definecolor{cycle8}{RGB}{247, 129, 191}

\newcommand*{\bigO}{\mathcal{O}}

\newcommand{\cmark}{\textcolor{cycle3}{\ding{52}}} %
\newcommand{\xmark}{\textcolor{cycle1}{\ding{56}}}
\newcommand{\mehmark}{\mbox{\cmark\textsubscript{\kern-0.45em\tiny\xmark}}}

\title{VERSE: Versatile Graph Embeddings from Similarity Measures}

\newcommand{\mpara}[1]{\medskip\noindent{\bf #1}}

\newcommand{\deepwalk}{\textsc{DeepWalk}\xspace}

\newcommand{\lineemb}{\textsc{LINE}\xspace}
\newcommand{\verseemb}{\textsc{VERSE}\xspace}
\newcommand{\fverseemb}{\textsc{fVERSE}\xspace}
\newcommand{\hypverseemb}{\textsc{hsVERSE}\xspace}
\newcommand{\grarep}{\textsc{GraRep}\xspace}
\newcommand{\nodetovec}{\textsc{Node2vec}\xspace}
\newcommand{\structovec}{\textsc{Struc2vec}\xspace}
\newcommand{\wordtovec}{\textsc{word2vec}\xspace}
\newcommand{\sdne}{\textsc{SDNE}\xspace}
\newcommand{\dngr}{\textsc{DNGR}\xspace}
\newcommand{\hoppe}{\textsc{HOPE}\xspace}
\newcommand*{\dblogcatalog}{\textsf{BlogCatalog}\xspace}
\newcommand*{\coauthor}{\textsf{CoAuthor}\xspace}
\newcommand*{\cocitation}{\textsf{CoCit}\xspace}
\newcommand*{\dyoutube}{\textsf{YouTube}\xspace}
\newcommand*{\dvk}{\textsf{VK}\xspace}

\newcommand*{\dorkut}{\textsf{Orkut}\xspace}

\renewcommand{\G}{\mathrm{G}}
\newcommand{\E}{\mathrm{E}}
\newcommand{\simg}{{sim}_\G}
\newcommand{\sime}{{sim}_\E}

\newcommand{\simadj}{\simg^{\mathrm{ADJ}}}

\newcommand{\simsr}{\simg^{\mathrm{SR}}}
\newcommand{\simppr}{\simg^{\mathrm{PPR}}}
\renewcommand{\L}{\mathcal{L}}
\renewcommand{\P}{\mathcal{P}}
\newcommand{\Q}{\mathcal{Q}}
\newcommand{\KL}{\mathrm{KL}}
\newcommand{\noisev}{\widetilde{v}}

\begin{document}

\setlength{\belowdisplayskip}{1pt} \setlength{\belowdisplayshortskip}{1pt}
\setlength{\abovedisplayskip}{1pt} \setlength{\abovedisplayshortskip}{1pt}

\author{Anton Tsitsulin}
\affiliation{%
  \institution{Hasso Plattner Institute}
}
\email{anton.tsitsulin@hpi.de}
\author{Davide Mottin}
\affiliation{%
  \institution{Hasso Plattner Institute}
}
\email{davide.mottin@hpi.de}
\author{Panagiotis Karras}
\affiliation{%
  \institution{Aarhus University}
}
\email{panos@cs.au.dk}
\author{Emmanuel M\"uller}
\affiliation{%
  \institution{Hasso Plattner Institute}
}
\email{emmanuel.mueller@hpi.de}

\begin{abstract}
{\em Embedding\/} a web-scale information network into a low-dimensional vector space facilitates tasks such as link prediction, classification, and visualization.
Past research has addressed the problem of extracting such embeddings by adopting methods from words to graphs, without defining a clearly comprehensible graph-related objective.
Yet, as we show, the objectives used in past works implicitly utilize similarity measures among graph nodes.
In this paper, we carry the similarity orientation of previous works to its logical conclusion; we propose VERtex Similarity Embeddings (\verseemb{}), a simple, versatile, and memory-efficient method that derives graph embeddings explicitly calibrated to preserve the distributions of a {selected} vertex-to-vertex similarity measure.
\verseemb{} {learns such embeddings by training a single\hyp{}layer neural network. While its default, {\em scalable\/} version does so via sampling similarity information, we also develop a {variant using the \emph{full} information per vertex}.}
Our experimental study on standard benchmarks and real-world datasets demonstrates that \verseemb{}, instantiated with diverse similarity measures, outperforms state-of-the-art methods in terms of precision and recall in major data mining tasks and supersedes them in time and space efficiency, while the scalable sampling-based variant achieves equally good results as the non-scalable full variant.
\end{abstract}

 \maketitle

\section{Introduction}\label{sec:introduction}

Graph data naturally arises in many domains, including social networks, protein networks, and the web.
Over the past years, numerous graph mining techniques have been proposed to analyze and explore such real-world networks.
Commonly, such techniques apply machine learning to address tasks such as node classification, link prediction, anomaly detection, and node clustering.
\begin{figure}[t!]
    \centering
    \subfloat[Community structure]{\resizebox{0.33\columnwidth}{!}{%
\begin{tikzpicture}
    \node (n1) at (0,0) [circle,draw,
        minimum size=.75cm,
        inner sep=0pt,
        fill=white!50!cycle1
        ]{};%
    \node (n2) at (2,1) [circle,draw,
        minimum size=.75cm,
        inner sep=0pt,
        fill=white!50!cycle1
        ]{};%
    \node (n3) at (-1,2) [circle,draw,
        minimum size=.75cm,
        inner sep=0pt,
        fill=white!50!cycle1
        ]{};
    \node (n4) at (1,3) [circle,draw,
        minimum size=.75cm,
        inner sep=0pt,
        fill=white!50!cycle1
        ]{};%
    \node (n5) at (2,6) [rectangle,draw,
        minimum size=.66467cm,
        inner sep=0pt,
        fill=white!50!cycle2
        ]{};
    \node (n6) at (4,7) [rectangle,draw,
        minimum size=.66467cm,
        inner sep=0pt,
        fill=white!50!cycle2
        ]{};
    \node (n7) at (1,8) [rectangle,draw,
        minimum size=.66467cm,
        inner sep=0pt,
        fill=white!50!cycle2
        ]{};
    \node (n8) at (3,9) [rectangle,draw,
        minimum size=.66467cm,
        inner sep=0pt,
        fill=white!50!cycle2
        ]{};
    \draw[-] (n1) -- (n2);
    \draw[-] (n1) -- (n3);
    \draw[-] (n2) -- (n3);
    \draw[-] (n2) -- (n4);
    \draw[-] (n3) -- (n4);
    \draw[-] (n4) -- (n5);
    \draw[-] (n5) -- (n6);
    \draw[-] (n5) -- (n7);
    \draw[-] (n6) -- (n7);
    \draw[-] (n6) -- (n8);
    \draw[-] (n7) -- (n8);
\end{tikzpicture} %
}}\label{sfig:motivation-comm}%
    \hfill
    \subfloat[Roles]{\resizebox{0.33\columnwidth}{!}{%
\begin{tikzpicture}
    \node (n1) at (0,0) [circle,draw,
        minimum size=.75cm,
        inner sep=0pt,
        fill=white!50!cycle1
        ]{};%
    \node (n2) at (2,1) [circle,draw,
        minimum size=.75cm,
        inner sep=0pt,
        fill=white!50!cycle1
        ]{};%
    \node (n3) at (-1,2) [circle,draw,
        minimum size=.75cm,
        inner sep=0pt,
        fill=white!50!cycle1
        ]{};
    \node (n4) at (1,3) [rectangle,draw,
        minimum size=.66467cm,
        inner sep=0pt,
        fill=white!50!cycle2
        ]{};%
    \node (n5) at (2,6) [rectangle,draw,
        minimum size=.66467cm,
        inner sep=0pt,
        fill=white!50!cycle2
        ]{};
    \node (n6) at (4,7) [circle,draw,
        minimum size=.75cm,
        inner sep=0pt,
        fill=white!50!cycle1
        ]{};
    \node (n7) at (1,8) [circle,draw,
        minimum size=.75cm,
        inner sep=0pt,
        fill=white!50!cycle1
        ]{};
    \node (n8) at (3,9) [circle,draw,
        minimum size=.75cm,
        inner sep=0pt,
        fill=white!50!cycle1
        ]{};
    \draw[-] (n1) -- (n2);
    \draw[-] (n1) -- (n3);
    \draw[-] (n2) -- (n3);
    \draw[-] (n2) -- (n4);
    \draw[-] (n3) -- (n4);
    \draw[-] (n4) -- (n5);
    \draw[-] (n5) -- (n6);
    \draw[-] (n5) -- (n7);
    \draw[-] (n6) -- (n7);
    \draw[-] (n6) -- (n8);
    \draw[-] (n7) -- (n8);
\end{tikzpicture} %
}}\label{sfig:motivation-roles}%
    \hfill
    \subfloat[Structural equivalence]{\resizebox{0.33\columnwidth}{!}{%
\begin{tikzpicture}
    \node (n1) at (0,0) [circle,draw,
        minimum size=.75cm,
        inner sep=0pt,
        fill=white!50!cycle1
        ]{};%
    \node (n2) at (2,1) [rectangle,draw,
        minimum size=.66467cm,
        inner sep=0pt,
        fill=white!50!cycle2
        ]{};%
    \node (n3) at (-1,2) [rectangle,draw,
        minimum size=.66467cm,
        inner sep=0pt,
        fill=white!50!cycle2
        ]{};
    \node (n4) at (1,3) [regular polygon,regular polygon sides=5,draw,
        minimum size=.862111cm,
        inner sep=0pt,
        fill=white!50!cycle3
        ]{};%
    \node (n5) at (2,6) [regular polygon,regular polygon sides=5,draw,
        minimum size=.862111cm, %
        inner sep=0pt,
        fill=white!50!cycle3
        ]{};
    \node (n6) at (4,7) [rectangle,draw,
        minimum size=.66467cm, %
        inner sep=0pt,
        fill=white!50!cycle2
        ]{};
    \node (n7) at (1,8) [rectangle,draw,
        minimum size=.66467cm,
        inner sep=0pt,
        fill=white!50!cycle2
        ]{};
    \node (n8) at (3,9) [circle,draw,
        minimum size=.75cm,
        inner sep=0pt,
        fill=white!50!cycle1
        ]{};
    \draw[-] (n1) -- (n2);
    \draw[-] (n1) -- (n3);
    \draw[-] (n2) -- (n3);
    \draw[-] (n2) -- (n4);
    \draw[-] (n3) -- (n4);
    \draw[-] (n4) -- (n5);
    \draw[-] (n5) -- (n6);
    \draw[-] (n5) -- (n7);
    \draw[-] (n6) -- (n7);
    \draw[-] (n6) -- (n8);
    \draw[-] (n7) -- (n8);
\end{tikzpicture} %
}}\label{sfig:motivation-structure}%
    \caption{Three node properties are highlighted on the same graph. Can a single model capture these properties?}\label{fig:new-motivation}
\end{figure}
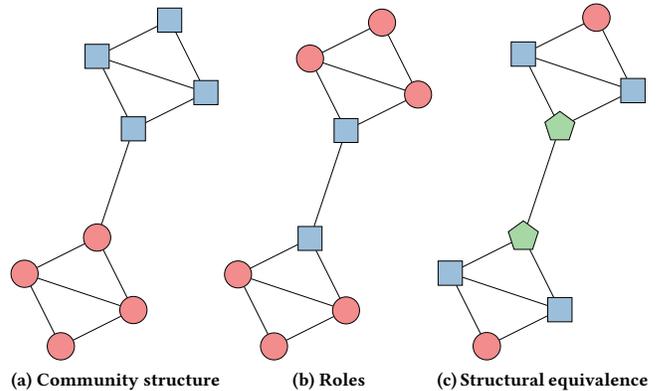 
Machine learning algorithms require a set of expressive discriminant features to characterize graph nodes and edges.
To this end, one can use features representing {\em similarities\/} among nodes~\cite{henderson2011}.
However, feature engineering is tedious work, and the results do not translate well across tasks~\cite{grover16}.
An alternative to feature design is {\em to learn\/} feature vectors, or {\em embeddings\/} by solving an optimization problem in {\em unsupervised\/} fashion.
Yet devising and solving a universal and tractable optimization problem for learning representations has withstood research efforts~\cite{bengio2013}.
One line of research~\cite{tang09,cao15} applies classical dimensionality reduction methods, such as SVD, to similarity matrices over the graph; yet these methods are burdened with constructing the matrix.
While a recent approach~\cite{ou2016} overcomes this impediment, it results in poor quality in prediction tasks due to its {\em linear nature}.

Another line of research aims to generate features capturing neighborhood locality, usually through an objective that can be optimized by Stochastic Gradient Descent (SGD)~\cite{perozzi14,tang15}.
Such methods rely on an implicit, albeit rigid, notion of node neighborhood; yet this one-size-fits-all approach cannot grapple with the diversity of real-world networks and applications.
Grover~et~al.~\cite{grover16} discerned this inflexibility in the notion of the local neighborhood; to ameliorate it, they proposed \nodetovec, which biases the exploration strategy of~\cite{perozzi14} using two hyperparameters.
Yet this {\em hyperparameter\hyp{}tuned\/} approach raises a cubic worst-case space complexity and compels the user to traverse several feature sets and gauge the one that attains the best performance in the downstream task.
Besides, a local neighborhood, even when found by hyperparameter tuning, still represents only one locality-based class of features; hence, {\nodetovec{}} does not adequately escape the rigidity it tries to mend.

We argue that features extracted by a more {\em versatile\/} similarity notion than that of a local neighborhood would achieve the flexibility to solve diverse {data mining} tasks in a large variety of graphs.
Figure~\ref{fig:new-motivation} makes a case for such a versatile similarity notion by exposing three distinct kinds of similarity on a graph: {\em community structure\/} guides community detection tasks, {\em roles\/} are typically used in classification, while {\em structural equivalence\/} defines peer correspondences in knowledge graphs.
As real-world tasks rely on a mix of such properties, a versatile feature learning algorithm should be capable of capturing all such similarities.

In this paper, we propose \verseemb, the first, to our knowledge, versatile graph embedding method that explicitly learns {\em any\/} similarity measures among nodes.
In its learning core, \verseemb{} stands between deep learning approaches~\cite{cao16,wang16} on the one hand and the direct decomposition of the similarity matrix~\cite{tang09,cao15} on the other hand.
Instead, \verseemb{} trains a simple, yet expressive, single-layer neural network to reconstruct similarity distributions between nodes.
Thereby, it outperforms previous methods in terms of both runtime and quality on a variety of large real networks and tasks.

Thanks to its ability to choose any appropriate similarity measure for the task at hand, \verseemb{} adjusts to that task without needing to change its core. Thereby, it fully ameliorates the rigidity observed in~\cite{grover16}, and integrates representation learning with feature engineering: any similarity measure, including those developed in feature engineering, can be used as input to \verseemb. For the sake of illustration, we instantiate our generic methodology using three popular similarity measures, namely Personalized PageRank (PPR)~\cite{page1999}, SimRank~\cite{jeh2002}, and adjacency similarity. We also show that versatility {\em does not\/} imply a new burden to the user, merely substituting hyperparameter tuning with similarity measure tuning: using PPR as a {\em default\/} choice for the similarity measure leads to good performance in nearly all tasks and networks we examined.

We summarize our contributions as follows. 
\begin{itemize}
	\item We propose a versatile %
    framework for graph embeddings that explicitly learns the distribution of \textbf{any} vertex similarity measure for each graph vertex.
    \item We %
    {interpret previous graph embeddings {through the lens of} our similarity framework, and instantiate \verseemb{} with Personalized PageRank, SimRank, and Adjacency similarity.} %
    \item We devise an efficient algorithm, linear in graph size, based on a single-layer neural network minimizing the divergence from real to reconstructed similarity distributions.
    \item In a thorough experimental evaluation, we show that \verseemb{} outperforms the state-of-the-art approaches in various graph mining tasks in quality while being even more efficient. %
\end{itemize}
\section{Related Work}\label{sec:related-work}

In the absence of a general-purpose representation for graphs, graph analysis tasks require domain experts to craft features~\cite{henderson2011,akoglu2010} or to use specialized feature selection algorithms~\cite{tangj2012,perozzi2014f}.
{Recently, specialized methods were introduced to learn representations of different graph parts~\cite{narayanan2016, abu2017} and graphs with annotations on nodes~\cite{huang2017, zhanga2017}, or edges~\cite{wang2017, hu2017}.
We focus on learning representations of \emph{nodes} in graphs without \emph{any} prior or additional information other than graph structure.}

\mpara{Traditional feature learning} learns features by compressing representations such as the Laplacian or adjacency matrix to a low-dimensional space.
Early works in this area include spectral techniques~\cite{belkin2001} and nonlinear dimensionality reduction~\cite{roweis2000, tenenbaum2000}.
In another vein, Marginal Fisher Analysis~\cite{yan2007} analyzes the dimensionality reduction of a point data set as the embedding of a graph capturing its statistic and geometric properties.
Such methods cannot be applied to large graphs, as they operate on dense matrices.

Some efforts have been made to overcome this limitation using enhanced linear algebra tools.
Ahmed~et~al.~\cite{ahmed2013} adopt stochastic gradient optimization for fast adjacency matrix eigendecomposition; Ou~et~al.~\cite{ou2016} utilize sparse generalized SVD to generate a graph embedding, \hoppe, from a similarity matrix amenable to decomposition into two sparse proximity matrices.
\hoppe{} is the first to support diverse similarity measures; however, it still requires the entire graph matrix as input and views the problem as one of linear dimensionality reduction rather than as one of nonlinear learning.
This way, it deviates not only from current research on graph embeddings but also from older works~\cite{yan2007}.

\mpara{Neural methods for representation learning.}
Advances in machine learning have led to the adoption of neural methods for learning representations~\cite{bengio2013}.
Building on the success of deep learning in domains such as image processing~\cite{krizhevsky2012} and Natural Language Processing (NLP)~\cite{bengio2003, mikolov13, pennington2014}, \wordtovec~\cite{mikolov13} builds word embeddings by training a single-layer neural network to guess the contextual words of a given word in a text.
Likewise, GloVe~\cite{pennington2014} learns a word space through a stochastic version of SVD in a transformed cooccurrence matrix.
While such text-based methods inherently take neighbor relationships into account, they require conceptual adaptations to model graphs~\cite{perozzi14}.

\mpara{Neural Graph Embeddings.}
The success of neural word embeddings inspired natural extensions towards learning graph representations~\cite{tu16, grover16, wang16, cao15, cao16, perozzi14}.
\deepwalk{}~\cite{perozzi14} first proposed to learn latent representations in a low\hyp{}dimensional vector space exploiting {\em local\/} node neighborhoods.
It runs a series of random walks of fixed length from each vertex and creates a matrix of \(d\)-dimensional vertex representations using the SkipGram algorithm of~\cite{mikolov13}.
These representations maximize the posterior probability of observing a neighboring vertex in a random walk.
\deepwalk{} embeddings can inform classification tasks using a simple linear classifier such as logistic regression.
\grarep~\cite{cao15} suggests using Singular Value Decomposition (SVD) on a log-transformed \deepwalk{} transition probability matrix of different orders, and then concatenate the resulting representations.
\structovec{}~\cite{ribeiro2017} rewires the graph to reflect isomorphism among nodes and capture structural similarities, and then derives an embedding relying on the \deepwalk{} core.
Works such as {\mbox{\cite{wang16,cao16}}} investigate deep learning approaches for graph embeddings.
Their results amount to complex models that require elaborate parameter tuning and computationally expensive optimization, leading to time and space complexities unsuitable for large graphs.

Nevertheless, all \deepwalk-based approaches use objective functions that are not tailored to graph structures. Some works~\cite{tang15,grover16,ribeiro2017} try to infuse graph-native principles into the learning process.
\lineemb{}~\cite{tang15} proposed graph embeddings that capture more elaborate proximity notions.
However, even \lineemb's notion of proximity is restricted to the immediate neighborhoods of each node; that is insufficient to capture the complete palette of node relationships~\cite{ou2016, grover16, ribeiro2017}.
Furthermore, \nodetovec~\cite{grover16} introduces two hyperparameters to regulate the generation of random walks and thereby tailor the learning process to the graph at hand in {\em semi-supervised\/} fashion.
However, \nodetovec{} remains attached to the goal of preserving local neighborhoods and requires laborious tuning for each dataset and each task.

\begin{table}
\begin{center}
{
\newcolumntype{C}{>{\centering\arraybackslash}X}
\begin{tabularx}{\columnwidth}{XCCCCC}
\multicolumn{1}{c}{} & \multicolumn{3}{c}{\textbf{Algorithm}} & \multicolumn{2}{c}{\textbf{Similarity}} \\
\cmidrule(lr){2-4}\cmidrule(lr){5-6}
\emph{method} & Local\hspace*{-5mm} & Scalable & Nonlinear & Global & Versatile \\ 
\midrule
DeepWalk~\cite{perozzi14}                & \multirow{2}{*}{\cmark{}}\hspace*{-5mm} & \multirow{2}{*}{\cmark{}} & \multirow{2}{*}{\cmark{}} & \multirow{2}{*}{\cmark{}} & \multirow{2}{*}{\xmark{}} \\
Node2vec~\cite{grover16} \\ 
LINE~\cite{tang15}                      & \cmark{}\hspace*{-5mm} & \cmark{} & \cmark{} & \xmark{} & \xmark{} \\
GraRep~\cite{cao15}                    & \xmark{}\hspace*{-5mm} & \xmark{} & \cmark{} & \cmark{} & \xmark{} \\
\sdne{}~\cite{wang16}                     & \cmark{}\hspace*{-5mm} & \xmark{} & \cmark{} & \xmark{} & \xmark{} \\
\dngr{} ~\cite{cao16}                    & \xmark{}\hspace*{-5mm} & \xmark{} & \cmark{} & \cmark{} & \xmark{} \\ %
\hoppe{}~\cite{ou2016}                  & \xmark{}\hspace*{-5mm} & \cmark{} & \xmark{} & \cmark{} & \cmark{} \\
\midrule
\verseemb{}               & \cmark{}\hspace*{-5mm} & \cmark{} & \cmark{}  & \cmark{}& \cmark{} \\
\bottomrule
\end{tabularx}
}
\end{center}
\caption{Outline of related work in terms of fulfilled (\cmark) and missing (\xmark) properties of algorithm and similarity measure.}\label{tbl:relatedwork}
\vspace{-7mm}
\end{table}

\mpara{{Overview.}} 
{%
Table~\ref{tbl:relatedwork} outlines five desirable properties for a graph embedding, and the extent to which previous methods {possess} them.
We distinguish between properties of \emph{algorithms}, {on the one hand}, and those of any {implicit or explicit} %
\emph{similarity measure} among nodes a method may express, {on the other hand}.
}

\begin{itemize}[leftmargin=0.1cm,itemindent=.3cm,labelwidth=\itemindent,labelsep=0cm,align=left]
\item {\bf local}: not requiring the entire graph matrix as input; \grarep{}, \dngr{}, and \hoppe{} fail in this respect.
\item {\bf scalable}: capable to process graphs with more than \(10^6\) nodes in less than a day; some methods fail in this criterion due to the dense matrix (\grarep{}), deep learning computations (\sdne{}), or both (\dngr{}).
\item {\bf nonlinear}: employing nonlinear transformations; \hoppe{} relies on a linear dimensionality reduction method, SVD\@; that is detrimental to its performance on building graph representations, just like linear dimensionality reduction methods fail to confer the advantages of their nonlinear counterparts in general~\cite{lee07book}.
\item {\bf global}: capable to model relationships between {\em any\/} pair of nodes; \lineemb{} and \sdne{} do not share this property as they fail to look beyond a node's immediate neighborhood. %
\item {\bf versatile}: supporting diverse similarity functions; \hoppe{} does so, yet is compromised by its {\em linear\/} character.
\end{itemize}

\section{Versatile Graph Embedding}\label{sec:verse-framework}

\verseemb{} possesses all properties mentioned {in our taxonomy}; it employs {\em nonlinear\/} transformation, desirable for dimensionality reduction~\cite{lee07book}; it is {\em local\/} in terms of the input it requires per node, but {\em global\/} in terms of the potential provenance of that input; it is {\em scalable\/} as it is based on sampling, and {\em versatile\/} by virtue of its generality.

\subsection{VERSE Objective}\label{subsec:problem-definition}

Given a graph \(\G = (V, E)\), where \(V = (v_1, \ldots, v_n)\), \(n = |V|\), is the set of vertices and \(E \subseteq (V \times V)\) the set of edges, we aim to learn a {\em nonlinear representation\/} of vertices \(v \in V\) to \(d\)-dimensional embeddings, where \(d \ll n\). 
Such representation is encoded into a \(n \times d\) matrix \(W\); the embedding of a node \(v\) is the row \(W_{v,\cdot}\) in the matrix; we denote it as \(W_{v}\) for compactness.

Our embeddings reflect {\em distributions\/} of a given graph similarity \(\simg : V \times V \rightarrow \mathbb{R}\) for every node \(v \in V\).
As such, we require that the similarities from any vertex \(v\) to all other vertices \(\simg(v,\cdot)\) are amenable to be interpreted as a distribution with \(\sum_{u \in V}{\simg(v,u)} = 1\) for all \(v \in V\).
We aim to devise $W$ by a scalable method that requires neither the \(V \times V\) {stochastic} similarity matrix nor its explicit materialization.

The corresponding node-to-node similarity in the embedded space is $\sime : V \times V \rightarrow \mathbb{R}$. As an optimization objective, we aim to minimize the Kullback-Leibler (KL) divergence from the given similarity distribution \(\simg{}\) to that of \(\sime{}\) in the embedded space:

\begin{equation}\label{eq:kl-objective}
\sum_{v \in V} \KL\left(\simg(v,\cdot) \;||\; \sime(v, \cdot)\right)
\end{equation}

We illustrate the usefulness of this objective using a small similarity matrix.
Figure~\ref{fig:karate-matrix} shows (a) the Personalized PageRank matrix, (b) the reconstruction of the same matrix by \verseemb, and (c) the reconstruction of the same matrix using SVD. %
It is visible that the nonlinear minimization of KL divergence between distributions preserves most of the information in the original matrix, while the linear SVD-based reconstruction fails to differentiate some nodes.
\begin{figure}[t]
\centering
    \subfloat[Similarity]{\includegraphics{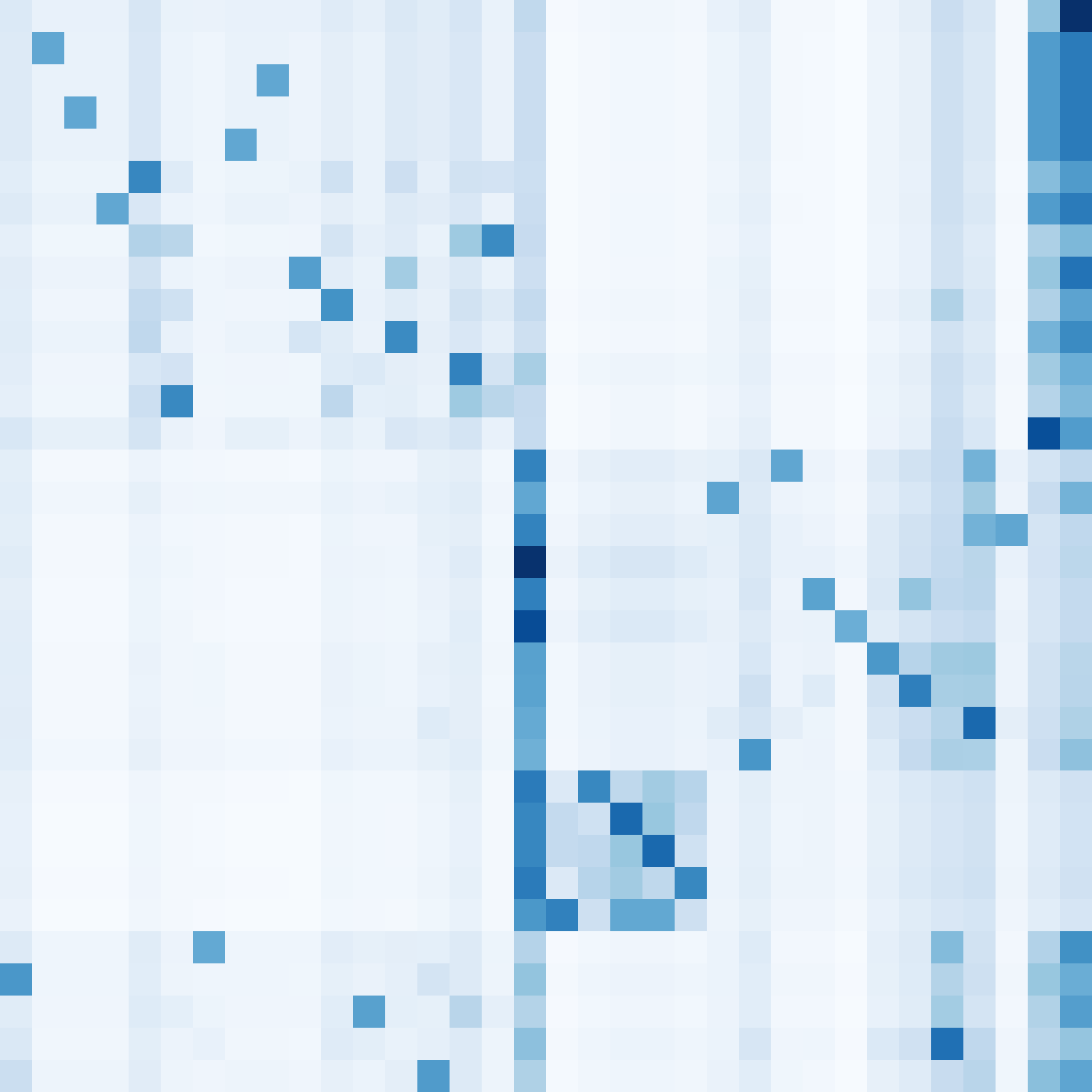}}\label{sfig:karate-real}
    \hfill
    \subfloat[\verseemb]{\includegraphics{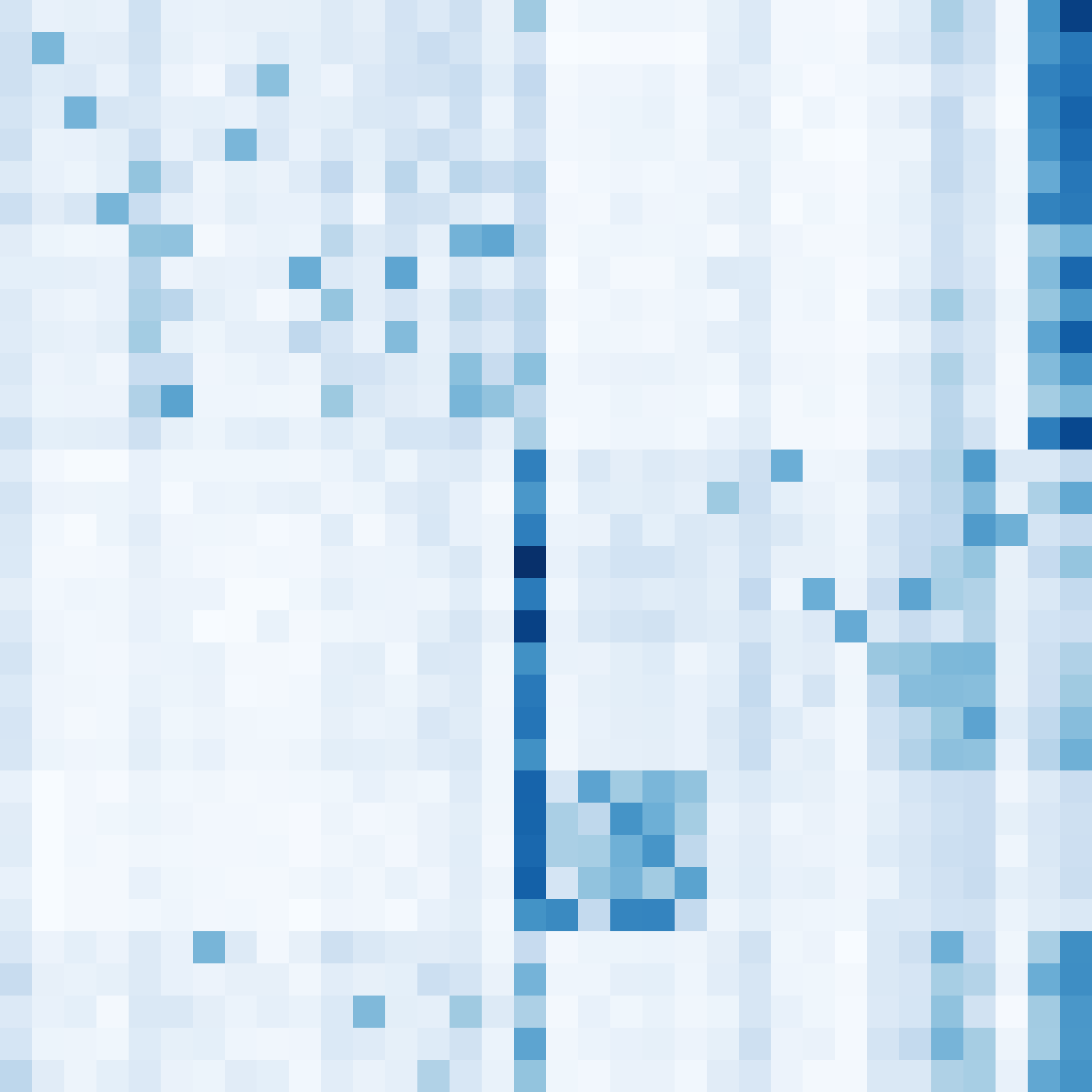}}\label{sfig:karate-verse}
    \hfill
    \subfloat[SVD]{\includegraphics{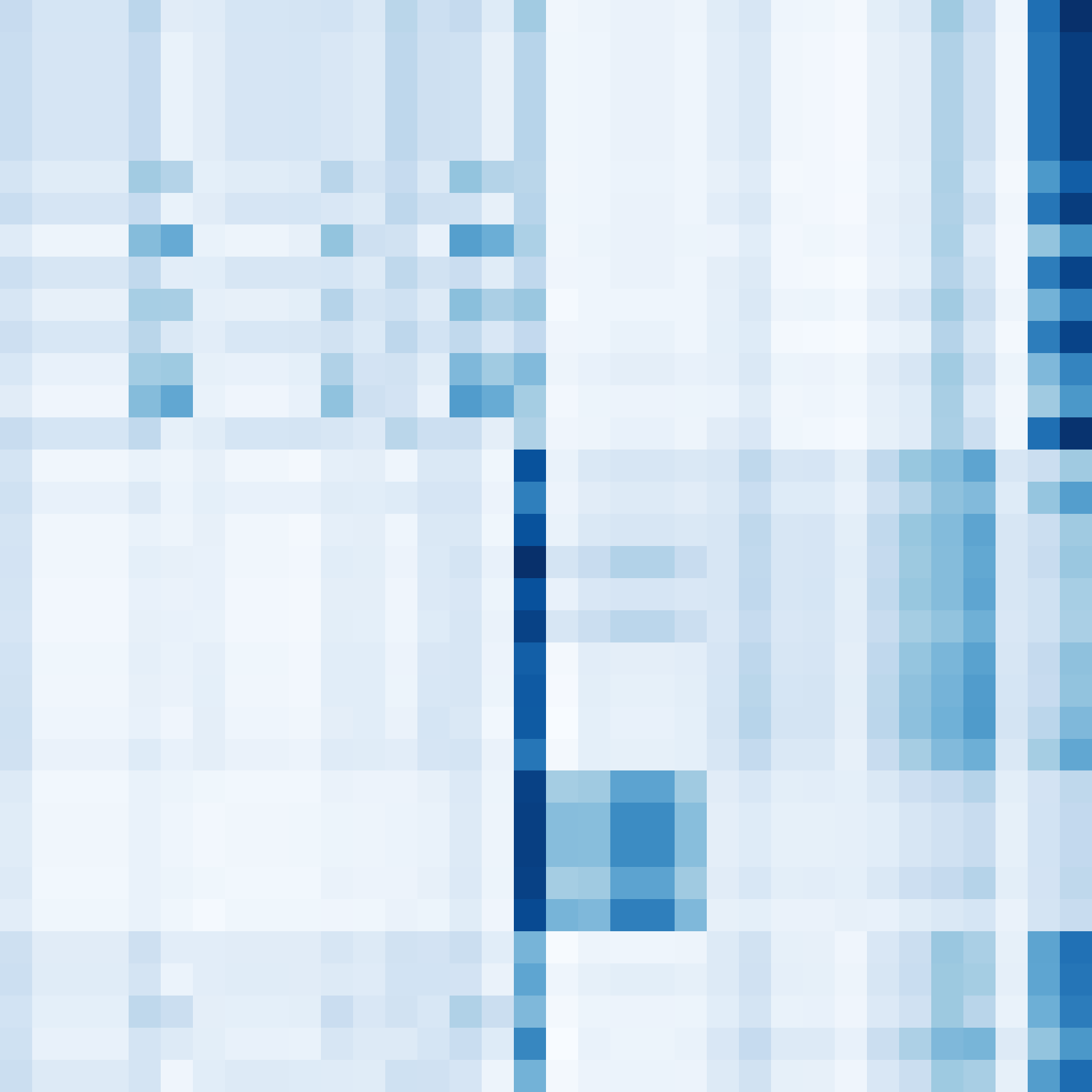}}\label{sfig:karate-svd}
\caption{An example similarity matrix and its reconstructions by \verseemb{} and SVD\@. Karate club graph~\protect\cite{zachary1977}, dimensionality \(d=4\) for both methods.\label{fig:karate-matrix}}
\end{figure} 

\subsection{VERSE Embedding Model}\label{subsec:model-optimization}

We define the unnormalized distance between two nodes \(u,v\) in the embedding space as the dot product of their embeddings \(W_u \cdot W_v^\top{}\).
The similarity distribution in the embedded space is then normalized with softmax:

\begin{equation}\label{eq:softmax-sime}
\sime(v, \cdot) = \frac{\exp(W_v W^{\top})} {\sum_{i=1}^{n}{\exp\left(W_v \cdot W_i\right)} }
\end{equation}

By Equation~\ref{eq:kl-objective}, we should minimize the \(\KL{}\)-divergence from \(\simg{}\) to \(\sime{}\); omitting parts dependent on \(\simg{}\) only, this objective is equivalent to minimizing the cross-entropy loss function~\cite{goodfellow2016}:

\begin{equation}\label{eq:cross-entropy}
\L = -\sum_{v \in V} {\simg(v,\cdot) \log\left(\sime\left(v,\cdot\right)\right) }
\end{equation}

We can accommodate this objective by stochastic gradient descent, which allows updating the model on each node singularly. However, a na{\"\i}ve version of gradient descent would require the full materialization of $\sime$ and $\simg$. Even in case $\simg$ is easy to compute on the fly, such as the adjacency matrix, the softmax in Equation~\ref{eq:softmax-sime} still has to be normalized over all nodes in the graph.

We use Noise Contrastive Estimation (NCE)~\cite{gutmann2010,mnih2012}, which allows us to learn a model that provably converges to its objective (see~\cite{gutmann2012},~Theorem~2).
NCE trains a binary classifier to distinguish between node samples coming from the empirical similarity distribution \(\simg{}\) and those generated by a noise distribution \(\Q{}\) over the nodes. Consider an auxiliary random variable \(D\) for node classification, such that \(D=1\) for a node drawn from the empirical distribution and \(D=0\) for a sample drawn from the noise distribution. Given a node \(u\) drawn from some distribution \(\P{}\) and a node \(v\) drawn from the distribution of \(\simg\left(u,\cdot\right)\), we draw \(s \ll n\) nodes \(\noisev{}\) from \(\Q(u)\) and use logistic regression to minimize the negative log-likelihood:

\vspace{-2mm}
\begin{equation}
\begin{aligned}\label{eq:nce}
\L_{NCE} = \sum_{\mathclap{\substack{
        u\sim \P \\
        v \sim \simg(u,\cdot)
}}}
\Big[ & \log \textstyle\Pr_W(D=1|\sime(u, v))+ \\
s \mathbb{E}_{\noisev \sim \Q(u)} & \log \textstyle\Pr_{W}(D=0|\sime(u, \noisev))\Big]
\end{aligned}
\end{equation}

\noindent
where \(\Pr_W\) is computed from $W$ as a sigmoid \(\sigma(x) = (1+e^{-x}){}^{-1}\) of the dot product between vectors $W_u$ and $W_v$, while we compute \(\sime\left(u,\cdot\right)\) {\em without\/} the normalization of Equation~\ref{eq:softmax-sime}. As the number of noise samples \(s\) increases, the NCE derivative provably converges to the gradient of cross-entropy~\cite{mnih2012}; %
thus, by virtue of NCE's asymptotic convergence guarantees, we are in effect minimizing the \(\KL{}\)-divergence from \(\simg{}\). NCE's theoretical guarantees depend on $s$, yet small values work well in practice~\cite{mnih2012}. In our experiments, we use $s = 3$.
These convergence guarantees of NCE are not affected by choice of distributions \(\P{}\) and \(\Q{}\)  (see~\cite{gutmann2012},~Corollary~5); however, its performance is empirically dependent on \(\Q{}\)~\cite{labeau2017}.

\subsection{Instantiations of \verseemb}\label{subsec:ppr-instantiation}

While \verseemb{} can be used with any similarity function, we choose to instantiate our model to widely used similarities $\simg$, namely Personalized PageRank (PPR), Adjacency Similarity, and SimRank. 

\mpara{Personalized PageRank.} %
Personalized PageRank~\cite{page1999} is a common similarity measure among nodes, practically used for many graph mining tasks~\cite{grover16,lu2011}.

\begin{definition}\label{def:ppr}
Given a starting node distribution \(s\), {\em damping factor\/} \(\alpha{}\), and the normalized adjacency matrix $A$, the \textbf{Personalized PageRank} vector \(\pi_s\) is defined by the recursive equation:

\vspace{-2mm}
$$\pi_s = \alpha s +(1-\alpha)\pi_{s}A$$
\end{definition}
\vspace{-2mm}

The stationary distribution of a random walk with restart with probability \(\alpha{}\) converges to PPR~\cite{page1999}. Thus, a sample from \(\simg(v, \cdot)\) is the last node in a single random walk from node \(v\). The damping factor \(\alpha{}\) controls the average size of the explored neighborhood. In Section~\ref{subsec:verse-relationships} we show that $\alpha$ is tightly coupled with the window size parameter $w$ of \deepwalk{} and \nodetovec.

\mpara{Adjacency similarity}.\label{ssec:adj-sim}
A straightforward similarity measure is the normalized adjacency matrix; this similarity corresponds to the \lineemb-1 model and takes into account only the immediate neighbors of each node. More formally, given the out degree \(Out(u)\) of node \(u\)

\vspace{-2mm}
\begin{align}\label{eq:adj-sim}
\simadj(u,v) = 
  \begin{cases} 
   1/Out(u) & \text{if } (u,v) \in E \\
   0       & \text{otherwise}
  \end{cases}
\end{align}

We experimentally demonstrate that \verseemb{} model is effective even in preserving the adjacency matrix of the graph.

\mpara{SimRank}.\label{ssec:simrank}
SimRank~\cite{jeh2002} is a measure of structural relatedness between two nodes, based on the assumption that two nodes are similar if they are connected to other similar nodes; SimRank is defined recursively as follows:

\vspace{-2mm}
\begin{equation}\label{eq:simrank-sim}
\simsr(u,v) = \frac{C}{\left|I(u)\right| \left|I(v)\right|}
 \sum_{i=1}^{\left|I(u)\right|}\sum_{j=1}^{\left|I(v)\right|}
 \simsr(I_i(u), I_j(v))
\end{equation}

where \(I(v)\) denotes the set of in-neighbors of node \(v\), and \(C\) is a number between \(0\) and \(1\) that geometrically discounts the importance of farther nodes. SimRank is a recursive procedure that involves computationally expensive operations: the straightforward method has the complexity of \(\bigO(n^4)\).%

SimRank values can be approximated up to a multiplicative factor dependent on $C$ through SimRank-Aware Random Walks (SARW)~\cite{jiang2017}.  
SARW computes a SimRank approximation through two reversed random walks with restart where the damping factor $\alpha$ is set to $\alpha = \sqrt{C}$. A reversed random walk traverses any edge \((u,v)\) in the opposite direction \((v,u)\). Since we are only interested in the distribution of each \(\simsr(v, \cdot)\), we can ignore the multiplicative factor in the approximation~\cite{jiang2017} that has little impact on our task.

\vspace{-2mm}
\begin{algorithm}[h]
    \begin{algorithmic}[1]
    \Function{VERSE}{$G, \simg{}, d$}
        \State{\(W \gets  \mathcal{N}\left(0,\,d^{-1}\right)\) \Comment{} \(\text{With}~W \in \mathbb{R}^{|V|\times d~}\)}
        \Repeat{}\label{alg:verse-loop}
            \State{\(u \sim \P{}\)} \Comment{} Sample a node
            \State{\(v \sim \simg(u)\)} \Comment{} Sample positive example
            \State{\(W_u, W_v \gets \textproc{Update}(u, v, 1)\)}
            \For{\(i \gets 1 \dotso s\)}
            \State{\(\noisev \sim \Q(u)\)} \Comment{} Sample negative example
            \State{\(W_u, W_\noisev \gets \textproc{Update}(u, \noisev, 0)\)}
            \EndFor{}
        
        \Until{converged}
        \State{\Return{\(W\)}}
    \EndFunction{}

    \Function{Update}{$u, v, D$} \Comment{} Logistic gradient update
    \State{\(g \gets (D - \sigma(W_u\cdot{}W_v))*\lambda{}\)}
    \State{\(W_u \gets g * W_v\)}
    \State{\(W_v \gets g * W_u\)}
    \EndFunction{}
    \end{algorithmic}
    \caption{\verseemb}\label{alg:verse}
\end{algorithm}
\vspace{-3mm} %
\subsection{VERSE Algorithm}\label{subsec:algorithm}

Algorithm~\ref{alg:verse} presents the overall flow of {\verseemb}.
Given a graph, a similarity function \(\simg{}\), and the embedding space dimensionality \(d\), we initialize the output embedding matrix \(W\) to \(\mathcal{N}(0, \frac{1}{d})\).
{Then, we  optimize our objective (Equation~\ref{eq:nce}) by gradient descent using {the} NCE algorithm discussed in the previous section. {To do so}, we repeatedly sample a node from the positive distribution \(\P{}\), sample the \(\simg{}\) (e.g. pick a neighboring node), and draw \(s\) negative examples.}
The \(\sigma{}\) in Line~13 represents the sigmoid function {\(\sigma=(1+e^{-x})^{-1}\)}, and \(\lambda{}\) the learning rate. 
We choose \(\P{}\) and \(\Q{}\) to be distributed uniformly by \(\mathcal{U}(1,n)\).

{As a strong baseline for applications {handling smaller} graphs, we also consider an elaborate, exhaustive variant of \verseemb, which computes {\em full similarity distribution\/} vectors per node instead of performing NCE-based sampling.
We name this variant \fverseemb{} and include it in our experimental study.}
Figure~\ref{fig:scalability-multiplot} presents our measures on the ability to reconstruct a similarity matrix for (i) \verseemb{} using NCE; (ii) \verseemb{} using Negative Sampling (NS) (also used in \nodetovec); and (ii) {the exhaustive \fverseemb{} variant.}
{We observe that, while} NCE approaches the exhaustive method in terms of matching the ground truth top-$100$ most similar nodes, NS fails to deliver the same quality.

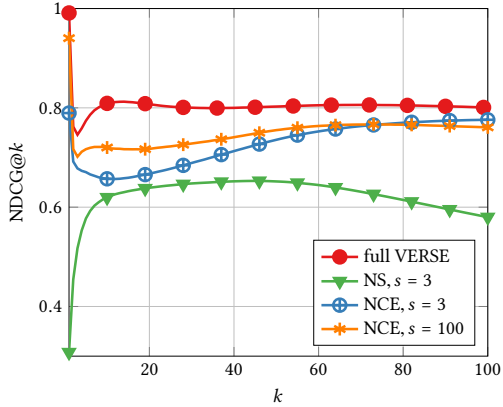
\begin{figure}[h]
\vspace{-3mm}
\resizebox{0.8\columnwidth}{!}{%
\begin{tikzpicture}
        \begin{axis}[                   						 
        grid=major,
        xlabel=\(k\), %
        ylabel=\textsc{NDCG}@\(k\),
        xmin=1,
        xmax=100,
        ymin=0.3,
        ymax=1,
        mark repeat={9},
        legend cell align=left,
        legend pos=south east
]
\addplot[very thick,color=cycle1,mark=*,mark size=3pt] coordinates {
(1, 0.9910977502)
(2, 0.76218008340155718)
(3, 0.74585663964000115)
(4, 0.75870278085057385)
(5, 0.77415446256467346)
(6, 0.7861357564846585)
(7, 0.79503473997336838)
(8, 0.80086170106630938)
(9, 0.80549924995920696)
(10, 0.80901917065289408)
(11, 0.81065127257342717)
(12, 0.81168191603158779)
(13, 0.81222670209635106)
(14, 0.81242759077186455)
(15, 0.81177362624862059)
(16, 0.81128263945889312)
(17, 0.81029992575936993)
(18, 0.80937157403319449)
(19, 0.80841511838423985)
(20, 0.80721142851767469)
(21, 0.80620523663056443)
(22, 0.80505527607014105)
(23, 0.80404938171237117)
(24, 0.80320371756717468)
(25, 0.80243916489491296)
(26, 0.8018466246409468)
(27, 0.80144434258182284)
(28, 0.80094213711812035)
(28, 0.80094213711812035)
(29, 0.80047667421442681)
(30, 0.80032671385750054)
(31, 0.80003752108374271)
(32, 0.79980012965484959)
(33, 0.79970699787350441)
(34, 0.79951082032826726)
(35, 0.79954729122733204)
(36, 0.79958893050112556)
(37, 0.79966152228729548)
(38, 0.79977206190969363)
(39, 0.79987122253003595)
(40, 0.79997630659226304)
(41, 0.80027183025004156)
(42, 0.80058156300929073)
(43, 0.80071368603325954)
(44, 0.80103231072020531)
(45, 0.80126561788362871)
(46, 0.8015464673456375)
(47, 0.80181459420678358)
(48, 0.80209122130650212)
(49, 0.80237022544262504)
(50, 0.80258751638351578)
(51, 0.80284629725166112)
(52, 0.80319184808360433)
(53, 0.80346842900891746)
(54, 0.8037478268238305)
(55, 0.80405791658824044)
(56, 0.80423162168074969)
(57, 0.80449699969523125)
(58, 0.8047686792250911)
(59, 0.80497629663397186)
(60, 0.80516674759849005)
(61, 0.80533378035148184)
(62, 0.80547863729062752)
(63, 0.80556248938211483)
(64, 0.80568298988643572)
(65, 0.80570193689809022)
(66, 0.80581934727556959)
(67, 0.80583548978130215)
(68, 0.8058618598485493)
(69, 0.80591358152257486)
(70, 0.80594347014404288)
(71, 0.80591116104366833)
(72, 0.80585877385853832)
(73, 0.80578984141047139)
(74, 0.80570663555738153)
(75, 0.8055882528446876)
(76, 0.80555415247152096)
(77, 0.80551268738720005)
(78, 0.80540600305196064)
(79, 0.80526430940283755)
(80, 0.80513118127928784)
(81, 0.80499354994836103)
(82, 0.80481441099340778)
(83, 0.80458981254992856)
(84, 0.80444340560992689)
(85, 0.8042326017780681)
(86, 0.80399323086207897)
(87, 0.80377664532404247)
(88, 0.80360458558957848)
(89, 0.80336005460103888)
(90, 0.80311701421359794)
(91, 0.8029034834191745)
(92, 0.80269535698526051)
(93, 0.80245734064939378)
(94, 0.80227839745385332)
(95, 0.80198909974602817)
(96, 0.8016795736622887)
(97, 0.80141480244183738)
(98, 0.80114683140303422)
(99, 0.80085313404889902)
(100, 0.80058516336962449)
};\addlegendentry{full \verseemb};\label{line:verse-ppr-blog2}
\addplot[very thick,color=cycle3,mark=triangle*,mark options={rotate=180},mark size=3pt] coordinates {
(1, 0.30769976726144299)
(2, 0.45431226309503819)
(3, 0.51733249308223905)
(4, 0.55311130191688773)
(5, 0.57735269466843531)
(6, 0.59245888983531869)
(7, 0.60298844507312377)
(8, 0.60984121832859317)
(9, 0.61562062494560699)
(10, 0.61962891106742513)
(11, 0.62316885297953151)
(12, 0.62620000156268618)
(13, 0.62838581670503701)
(14, 0.63018110116685611)
(15, 0.63214131922754002)
(16, 0.63376498895841238)
(17, 0.63515235991244856)
(18, 0.6366223926703366)
(19, 0.63783463201349488)
(20, 0.63915464194247917)
(21, 0.64013680763863623)
(22, 0.64118544679087386)
(23, 0.64219490773051147)
(24, 0.64315159168477809)
(25, 0.64393383771135948)
(26, 0.64475852058983363)
(27, 0.64562581471250358)
(28, 0.64639316150759119)
(29, 0.64711989909177936)
(30, 0.64776823320551113)
(31, 0.64827107926863292)
(32, 0.64879118461632079)
(33, 0.6492255379943801)
(34, 0.64963773573783457)
(35, 0.65016641593833657)
(36, 0.65063287233697242)
(37, 0.65100943745739548)
(38, 0.65144201473519014)
(39, 0.65175849970225974)
(40, 0.65207872140639855)
(41, 0.65230594462539415)
(42, 0.65255258307399455)
(43, 0.65271811280673253)
(44, 0.65281957953754366)
(45, 0.65289657166852422)
(46, 0.65280826296450112)
(47, 0.65266908515252009)
(48, 0.65250251877792953)
(49, 0.65217503733049387)
(50, 0.65182282534412472)
(51, 0.65145485504938905)
(52, 0.6510306517494493)
(53, 0.65046677327568037)
(54, 0.64975215063475311)
(55, 0.64901169872917119)
(56, 0.64819432140306321)
(57, 0.64730936229943226)
(58, 0.64649270916455348)
(59, 0.64555081833899919)
(60, 0.64455926978616085)
(61, 0.64345267883942026)
(62, 0.64224257093404513)
(63, 0.64093738455211247)
(64, 0.63964474115877623)
(65, 0.63833598395049085)
(66, 0.63698816837188144)
(67, 0.6355477801019368)
(68, 0.63409026073371433)
(69, 0.63258042382815904)
(70, 0.63106375892654354)
(71, 0.62947348181607454)
(72, 0.62787873989048337)
(73, 0.62629018435491257)
(74, 0.62459239391152876)
(75, 0.62296998487921873)
(76, 0.62135881993428699)
(77, 0.61965662084011341)
(78, 0.61797124941860859)
(79, 0.6162553700962301)
(80, 0.6145712527298649)
(81, 0.61286010513333156)
(82, 0.61111666955345156)
(83, 0.60936940480980051)
(84, 0.60761094153786321)
(85, 0.60588723766452079)
(86, 0.6041585878410215)
(87, 0.60239553484049868)
(88, 0.60060937788304902)
(89, 0.59887459825826295)
(90, 0.59710924512394647)
(91, 0.59536044552735279)
(92, 0.59360387925406843)
(93, 0.59184979104412994)
(94, 0.59009605186393632)
(95, 0.58837594993208353)
(96, 0.58668762325260837)
(97, 0.58501662474101646)
(98, 0.58329012573757522)
(99, 0.58156853606839232)
(100, 0.57989144455569075)
};\addlegendentry{$\textsc{NS}, s=3$};\label{line:sverse-ppr-neg-blog-3}
\addplot[very thick,color=cycle2,mark=oplus,mark size=3pt] coordinates {
(1, 0.78954615981380916)
(2, 0.69181007712444187)
(3, 0.67804752940101642)
(4, 0.67318338082709994)
(5, 0.67053657859758431)
(6, 0.66598010485516224)
(7, 0.66313346987500899)
(8, 0.66014062955834951)
(9, 0.65846609412477097)
(10, 0.65711571519243872)
(11, 0.65679114858532661)
(12, 0.65678549890859617)
(13, 0.65721499444027132)
(14, 0.65826909939051437)
(15, 0.65918785569011484)
(16, 0.66057846333175008)
(17, 0.66224325089651603)
(18, 0.66384390133612181)
(19, 0.66555835028580623)
(20, 0.66750282046181164)
(21, 0.66933978870389765)
(22, 0.67135830484260262)
(23, 0.67321742508956417)
(24, 0.67524480886821248)
(25, 0.67740942953779959)
(26, 0.67961231639963426)
(27, 0.68183135812935136)
(28, 0.68419989668246506)
(29, 0.68671802893835165)
(30, 0.68901069315991581)
(31, 0.6914005977354466)
(32, 0.69365754541933222)
(33, 0.69591033330242524)
(34, 0.69843465625045686)
(35, 0.70085351287552622)
(36, 0.70339621923989892)
(37, 0.70584011679783432)
(38, 0.70819765445594807)
(39, 0.71051425517048672)
(40, 0.71292179243111131)
(41, 0.71537543550496685)
(42, 0.71767941037218808)
(43, 0.7200911612014087)
(44, 0.72239585072823798)
(45, 0.72482978418048893)
(46, 0.72705526276477228)
(47, 0.72926141159752689)
(48, 0.73147565141533666)
(49, 0.73358136203103774)
(50, 0.7356711074070863)
(51, 0.73771196786250104)
(52, 0.73966554496358372)
(53, 0.74155061189194982)
(54, 0.74338631503838848)
(55, 0.74512410025576448)
(56, 0.74668173097594481)
(57, 0.74822298659016895)
(58, 0.74975069872830813)
(59, 0.75118207988175245)
(60, 0.75255146569784159)
(61, 0.75388154057297641)
(62, 0.75511680876790543)
(63, 0.75627121542691289)
(64, 0.75737267198852698)
(65, 0.75844929287284768)
(66, 0.75946786112065756)
(67, 0.76050822813995889)
(68, 0.76149904453835848)
(69, 0.76239883430283739)
(70, 0.76324834801124775)
(71, 0.7640570848795476)
(72, 0.76482478018191091)
(73, 0.76555445493127483)
(74, 0.76622924183311647)
(75, 0.76691213994974938)
(76, 0.76753871071268298)
(77, 0.76817496175883915)
(78, 0.76882634950103124)
(79, 0.76941395290305548)
(80, 0.77000874057614832)
(81, 0.77047314274504008)
(82, 0.77095587595724846)
(83, 0.77140297979740668)
(84, 0.77181378827378633)
(85, 0.77219186189069333)
(86, 0.77262167737904541)
(87, 0.77303532346968418)
(88, 0.77339306413496411)
(89, 0.77370637340587678)
(90, 0.77403817725619384)
(91, 0.77434216875679263)
(92, 0.77452566263283984)
(93, 0.77470962568990742)
(94, 0.77499897333739032)
(95, 0.77524352456356849)
(96, 0.77543925900621968)
(97, 0.77556725548464567)
(98, 0.77572943455121934)
(99, 0.77586159733492344)
(100, 0.77600468661255295)
};\addlegendentry{$\textsc{NCE}, s=3$};\label{line:sverse-ppr-blog-3}
\addplot[very thick,color=cycle5,mark=asterisk,mark size=3pt] coordinates {
(1, 0.94036074476338249)
(2, 0.71695379725074582)
(3, 0.70159471020849162)
(4, 0.71164873352220936)
(5, 0.7167713265540534)
(6, 0.71956248704316617)
(7, 0.72126431099427057)
(8, 0.72129429432262426)
(9, 0.72126891194288401)
(10, 0.72011318859188489)
(11, 0.71830121688249082)
(12, 0.71777780268752034)
(13, 0.71759599670866059)
(14, 0.71672729910634536)
(15, 0.71667933234299519)
(16, 0.71650867666911244)
(17, 0.71661815232927073)
(18, 0.71632294173424382)
(19, 0.71675307699781787)
(20, 0.71733629572926372)
(21, 0.71857546037037212)
(22, 0.71969742277226967)
(23, 0.72054401235799204)
(24, 0.72159505927428047)
(25, 0.72238000215697118)
(26, 0.72354975918170972)
(27, 0.72464311788107016)
(28, 0.72586707929118577)
(29, 0.72682859127589783)
(30, 0.72819190412334067)
(31, 0.72935765512086315)
(32, 0.73056528084006656)
(33, 0.73168196106563421)
(34, 0.73277195301962805)
(35, 0.73425304318022389)
(36, 0.73542878448316629)
(37, 0.7367531434519905)
(38, 0.73812692461379692)
(39, 0.73937333200038735)
(40, 0.74086476887237407)
(41, 0.74238055259586266)
(42, 0.74379579006938268)
(43, 0.74535974748979472)
(44, 0.74675123815405364)
(45, 0.74838225242336232)
(46, 0.74986669635161651)
(47, 0.75124967212678084)
(48, 0.7526396753115373)
(49, 0.75385431511607248)
(50, 0.75514270382060333)
(51, 0.75609455417250737)
(52, 0.75726819940646162)
(53, 0.75821040701693843)
(54, 0.75916045888733974)
(55, 0.75996969383572521)
(56, 0.7606333790584997)
(57, 0.76145276786482541)
(58, 0.7623106763812475)
(59, 0.76296675695172855)
(60, 0.7635008494726454)
(61, 0.76398834634735135)
(62, 0.76429899653260236)
(63, 0.76477390264706113)
(64, 0.76518567949290628)
(65, 0.76565171769687324)
(66, 0.76585756792768556)
(67, 0.76599607262297975)
(68, 0.76622758481697195)
(69, 0.7663531199453657)
(70, 0.76640128349046976)
(71, 0.76647302167889375)
(72, 0.76661935077994903)
(73, 0.76663906255448111)
(74, 0.76649150607107741)
(75, 0.76647869123237067)
(76, 0.76641605603066454)
(77, 0.76638793927775217)
(78, 0.76614516215138606)
(79, 0.76615985508849005)
(80, 0.76603080604631812)
(81, 0.76590631028260914)
(82, 0.76574665390479824)
(83, 0.76548788812793345)
(84, 0.76520911167880579)
(85, 0.76492779869310856)
(86, 0.76466464793951983)
(87, 0.76438643663727812)
(88, 0.76417090854673475)
(89, 0.76393453244469922)
(90, 0.76367031593537449)
(91, 0.76333123814869297)
(92, 0.76298469215505405)
(93, 0.76262128713104727)
(94, 0.76237398195646233)
(95, 0.7621062838711109)
(96, 0.76186301673946943)
(97, 0.76157835024487952)
(98, 0.76129042982365014)
(99, 0.76101903315862529)
(100, 0.76071300580369861)
};\addlegendentry{$\textsc{NCE}, s=100$};\label{line:sverse-ppr-blog-100}
\end{axis}
\end{tikzpicture} %
}%
\vspace{-3mm}
\caption{Ranking preformance in terms of NDCG for reconstructing PPR similarity, averaged across nodes in a graph.}\label{fig:scalability-multiplot}
\vspace{-2mm}
\end{figure}
\subsection{Complexity Comparison}\label{subsec:complexity}

Table~\ref{tbl:complexity} presents the average (\(\Theta{}\)) and worst-case (\(\bigO{}\)) time and space complexity of \verseemb, along with those of methods in previous works; $d$ is the embedding dimensionality, $n$ the number of nodes, $m$ the number of edges, and $s$ the number of samples used, and $t$ the number of iterations in \grarep.
Methods that rely on {fast} sampling (\verseemb{} and \lineemb) require time linear in $n$ and space quadratic in $n$ in the worst case. \deepwalk{} requires \(\bigO(n\log n)\) time due to its use of hierarchical softmax. \nodetovec{} stores the neighbors\hyp{}of\hyp{}a\hyp{}neighbor, incurring a quadratic cost in sparse graphs, but cubic in dense graphs. Thus, \verseemb{} comes at the low end of complexities compared to previous work on graph embeddings. 
{Remarkably, even the computationally expensive \fverseemb{} affords complexity comparable to some previous works.}
\begin{table}[!b]
\vspace{-2mm}
\begin{center}
{%
\newcolumntype{C}{>{\centering\arraybackslash}X}
\renewcommand{\aboverulesep}{0pt}
\renewcommand{\belowrulesep}{0pt}
\begin{tabularx}{\columnwidth}{XCCCC}
& \multicolumn{2}{c}{\textbf{Time}} & \multicolumn{2}{c}{\textbf{Space}} \\
\cmidrule(lr){2-3}\cmidrule(lr){4-5}
\emph{method} & \(\Theta{}\) & \(\bigO{}\) & \(\Theta{}\) & \(\bigO{}\) \\ 
\midrule
\deepwalk{}                  & \(dn\log n\) & \(dn\log n\) & \(m\) & \(n^2\) \\
\grarep{}                    & \(tn^3\) & \(tn^3\) & \(n^2\) & \(n^2\) \\
\lineemb{}                      & \(dsn\) & \(dsn\) & \(m\) & \(n^2\) \\
\nodetovec{}                  & \(dsn\) & \(dsn\) & \(\frac{m^2}{n}\) & \(n^3\) \\
\hoppe{}                      & \(d^2m\) & \(d^2n^2\) & \(m\) & \(n^2\) \\
\midrule
\rowcolor{cycle2!8}
\fverseemb{}               & \(dn^2\) & \(dn^2\) & \(n^2\) & \(n^2\) \\
\midrule
\rowcolor{cycle2!35}
\verseemb{}               & \(dsn\) & \(dsn\) & \(m\) & \(n^2\) \\
\bottomrule
\end{tabularx}
}
\end{center}
\caption{Comparison of neural embedding methods in terms of average (\(\Theta{}\)) and worst-case (\(\bigO{}\)) time and space complexity.}\label{tbl:complexity}
\vspace{-6mm}
\end{table} %

\subsection{Similarity Notions in Previous Approaches}\label{subsec:verse-relationships}

Here, we provide additional theoretical considerations of \verseemb{} compared to \lineemb~\cite{tang15}, \deepwalk~\cite{perozzi14} and \nodetovec~\cite{grover16} and demonstrate how our general model subsumes and extends previous research in versatility and scalability. 

\mpara{Comparison with \deepwalk{} and \nodetovec.}
\deepwalk{} and \nodetovec{} generate samples from random walks of fixed window size $w$ by the \wordtovec{} sampling strategy~\cite{mikolov13}. We derive a relationship between the window size $w$ of that strategy and the damping factor $\alpha$ of Personalized PageRank. 

\begin{lemma}
Let $X_r$ be the random variable that represents the length of a random walk $r$ sampled with parameter $w$ by the \wordtovec{} sampling strategy. Then for any $0 < j \leq w$

\begin{equation}\label{eq:dw-sim}
\Pr(X_r=j) = \frac{2}{w(w+1)}(w-j + 1) 
\end{equation}
\end{lemma}
\begin{proof}

For each node $v \in V$, \wordtovec{} strategy samples two random walks of length $w$ starting from $v \in V$. These two random walks represents the {\em context\/} of $v$, where $v$ is the central node of a walk of length $2w + 1$. 
The model is then trained on increasing context size up to $w$. 
Therefore, the number of nodes sampled for each random walk amount to $\sum_{i=1}^w i= \frac{w(w+1)}{2}$. A node at distance $0 < j \leq w$ is sampled $(w - j + 1)$ times; thus, the final probability is $\frac{2}{w(w+1)}(w-j + 1)$.
\end{proof}

Personalized PageRank provides the maximum likelihood estimation for the distribution in Equation~\ref{eq:dw-sim} for $\alpha = \frac{w-1}{w+1}$. Then, $w = 10$ corresponds to $\alpha = 0.82$, which is close to the standard $\alpha = 0.85$, proved effective in practice~\cite{brin1998anatomy}. 
On the other hand, $\alpha = 0.95$, which{, for example, achieves the best performance on a task} in Section~\ref{subsec:classification}, corresponds to $w = 39$. Such large $w$ prohibitively increases the computation time for \deepwalk{} and \nodetovec.

\mpara{Comparison with \lineemb.}
\lineemb{} introduces the concept of first- and second-order proximities to model complex node relationships. As we discussed, in \verseemb, first-order proximity corresponds to the dot-product among the similarity vectors in the embedding space:

\begin{equation*}
        \sime(u,v) = W_u \cdot W_v
\end{equation*}
\vspace{1mm}

On the other hand, second-order proximity corresponds to letting \verseemb{} learn one more matrix $W'$, so as to model {\em asymmetric\/} similarities of nodes in the embedding space. We do that by defining $\sime$ asymmetrically, using both $W$ and $W'$:

\begin{equation*}
        \sime(u,v) = W_u \cdot W'_v
\end{equation*}
\vspace{1mm}

The intuition behind second-order proximity is the same as that of SimRank: similar nodes have similar neighborhoods.
Every previous method, except for \lineemb-1, used second-order proximities, due to the \wordtovec{} interpretation of embeddings borrowed by \deepwalk{} and \nodetovec. In our model, second-order proximities can be encoded %
by adding an additional matrix; we empirically evaluate their effectiveness in Section~\ref{sec:experiments}.

\section{Experiments}\label{sec:experiments}

We evaluate \verseemb{} against several state-of-the-art graph embedding algorithms.
{For repeatability purposes, we provide all data sets and the C++ source code for \verseemb{}\footnote{\url{https://github.com/xgfs/verse}}, \deepwalk{}\footnote{\url{https://github.com/xgfs/deepwalk-c}} and \nodetovec{}\footnote{\url{https://github.com/xgfs/node2vec-c}}.}
We run the experiments on an Amazon AWS c4.8 instance with 60Gb RAM\@.
Each method is assessed on the best possible parameters, with early termination of the computation in case no result is returned within one day.
We provide the following state-of-the-art graph embedding methods for comparison:

\begin{itemize}[leftmargin=0cm,itemindent=.4cm,labelwidth=\itemindent,labelsep=0cm,align=left]
	\item \deepwalk~\cite{perozzi14}: This approach learns an embedding by sampling random walks from each node, applying \wordtovec-based learning on those walks. 
    We use the default parameters described in the paper, i.e., walk length \(t=80\), number of walks per node \(\gamma = 80\), and window size \(w=10\).
    \item \lineemb~\cite{tang15}: This approach learns a \(d\)-dimensional embedding in two steps {, both using adjacency similarity.
    First, it learns \(d/2\) dimensions {using} first-order proximity; then, it learns another \(d/2\) features {using} second-order proximity. }
    Last, the two halves are normalized and concatenated.
    We obtained a copy of the code\footnote{\url{https://github.com/tangjianpku/LINE}} and run experiments with total \(T=10^{10}\) samples and \(s=5\) negative samples, as described in the paper. 
    \item \grarep{}~\cite{cao15}: This method factorizes the full adjacency similarity matrix using SVD, multiplies the matrix by itself, and repeats the process $t$ times. The final embedding is obtained by concatenating the factorized vectors. We use $t=4$ and $32$ dimensions for each SVD factorization; thus, the final embedding has $d=128$.
    \item \hoppe~\cite{ou2016}: This method is a revised Singular Value Decomposition restricted to sparse similarity matrices. We report the results obtained running \hoppe{} with the default parameters, i.e, Katz similarity {(an extension of Katz centrality \cite{katz53})} {as the} similarity measure and {\(\beta{}\) inversely proportional to the spectral radius}. Since Katz similarity {does not converge} on directed graphs with sink nodes, we used Personalized PageRank with \(\alpha{} = 0.85\) for {the} \cocitation{} dataset.
    \item \nodetovec~\cite{grover16}: This is a {\em hyperparameter-supervised\/} approach that extends \deepwalk{} by adding two parameters, \(p\) and \(q\), so as to control \deepwalk's random walk sampling. %
    The special case with parameters \(p=1, q=1\) corresponds to \deepwalk{}; yet, sometimes \nodetovec{} shows worse performance than \deepwalk{} in our evaluation, due to the fact it uses negative sampling, while \deepwalk{} uses hierarchical softmax. %
    We fine-tuned the hyperparameters $p$ and $q$ on each dataset and task. Moreover, we used a large training data to fairly compare to \deepwalk, i.e., walk length \(l=80\), number of walks per node \(r = 80\), and window size \(w=10\).
\end{itemize}

\mpara{Baselines.} In addition to graph embeddings methods, we implemented the following baselines. 

\begin{itemize}[leftmargin=0cm,itemindent=.4cm,labelwidth=\itemindent,labelsep=0cm,align=left]
	\item Logistic regression: We use the well-known logistic regression method as a baseline for link prediction. 
    We train the model on a set of common node-specific features, namely node degree, number of common neighbors, Adamic-Adar, Jaccard coefficient, preferential attachment, and resource allocation index~\cite{lu2011,lichtenwalter2010}.
	\item Louvain community detection~\cite{blondel2008}: We employ a standard partition method for community detection as a baseline for graph clustering, reporting the best partition in terms of modularity~\cite{newman2006}.
\end{itemize}

\begin{table}
\setlength{\tabcolsep}{3.5pt}
\centering
{\small
\begin{tabular}{lcrcccr}
\multicolumn{1}{c}{} & \multicolumn{3}{c}{\textbf{Size}} & \multicolumn{3}{c}{\textbf{Statistics}} \\
\cmidrule(lr){2-4}\cmidrule(lr){5-7}
\emph{dataset} & \(|V|\) & \multicolumn{1}{c}{\(|E|\)} & \(|\mathcal{L}|\) & \mbox{Avg.~degree} & Mod. & Density \\
    \midrule
\dblogcatalog{} & 10k & 334k & 39 & 64.8 & 0.24 & \mbox{\(6.3 \times 10^{-3}\)} \\
\mbox{\cocitation{}} & 44k & 195k & 15 & 8.86 & 0.72 & \mbox{\(2.0 \times 10^{-4}\)} \\
\mbox{\coauthor{}} & 52k & 178k & --- & 6.94 & 0.84 & \mbox{\(1.3 \times 10^{-4}\)} \\
\dvk{} & 79k & 2.7M & 2 & 34.1 & 0.47 & \mbox{\(8.7 \times 10^{-4}\)} \\
\dyoutube{} & 1.1M & 3M & 47 & 5.25 & 0.71 & \mbox{\(9.2 \times 10^{-6}\)} \\ 
\dorkut{} & 3.1M & 234M & 50 & 70 & 0.68 & $2.4 \times 10^{-5}$ \\ 
\bottomrule
\end{tabular}
}
\caption{Dataset characteristics: number of vertices \(|V|\), number of edges \(|E|\); number of node labels \(|\mathcal{L}|\); average node degree; modularity~\protect\cite{newman2006}; density defined as \(|E|/\binom{|V|}{2}\).}\label{tbl:datasets}
\vspace{-23pt}
\end{table}

\mpara{Parameter settings.}
In line with previous research~\cite{perozzi14, tang15, grover16} we set the embedding dimensionality \(d\) to \(128\).
The learning procedure (Algorithm~\ref{alg:verse}, Line~\ref{alg:verse-loop}) is run \(10^5\) times for \verseemb{} and \(250\) times for \fverseemb{}; the difference in setting is motivated by the number of model updates which is \(\bigO(n)\) in \verseemb{} and \(\bigO(n^2)\) in \fverseemb.

We use LIBLINEAR~\cite{fan2008} to perform logistic regression with default parameter settings. 
Unlike previous work~\cite{perozzi14,tang15,grover16} we employ a stricter assumption for multi-label node classification: the number of correct classes is not known apriori, but found through the Label Powerset multi-label classification approach~\cite{tsoumakas2006}.

For link prediction and multi-label classification, we evaluated each individual embedding \(10\) times in order to reduce the noise introduced by the classifier. 
Unless otherwise stated, we run each experiment \(10\) times, and report the average value among the runs. 
Throughout our experimental study, we use the above parameters as default, unless indicated otherwise.

\mpara{Datasets.}
We test our methods on {six} real datasets; we report the main data characteristics in Table~\ref{tbl:datasets}.

\begin{itemize}[leftmargin=0cm,itemindent=.4cm,labelwidth=\itemindent,labelsep=0cm,align=left]
    \item {\dblogcatalog~\cite{zafarani2009} is a network of social interactions among bloggers in the \dblogcatalog{} website. Node-labels represent topic categories provided by authors.}
    \item Microsoft Academic Graph~\cite{MSAccGraph} is a network of academic papers, citations, authors, and affiliations from Microsoft Academic website released for the KDD-2016 cup. It contains 150 million papers up to February 2016 spanning various disciplines from math to biology. We extracted two separate subgraphs from the original network, using \(15\) conferences in data mining, databases, and machine learning. The first, \coauthor, is a co-authorship network among authors. The second, \cocitation, is a network of papers citing other papers; labels represent conferences in which papers were published.  
        \item \dvk{} is a Russian all-encompassing social network. We extracted two snapshots of the network in November 2016 and May 2017 to obtain information about link appearance. We use the gender of the user for classification and country for clustering.
    \item \dyoutube~\cite{tang2009} is a network of social interactions among users of the \dyoutube{} video platform. The labels represent groups of viewers by video genres.
    \item \dorkut~\cite{yang2015} is a network of social interactions among users of the \dorkut{} social network platform. The labels represent communities of users. We extracted the 50 biggest communities and use them as labels for classification.
\end{itemize}

\mpara{Evaluation methodology.} %
The default form of \verseemb{} runs Personalized PageRank with $\alpha = 0.85$.
For the sake of fairness, we design a {\em hyperparameter-supervised\/} variant of \verseemb, by analogy to the hyperparameter-tuned variant of \deepwalk{} introduced by \nodetovec~\cite{grover16}. This variant, \hypverseemb, selects the best similarity with cross-validation across two proximity orders (as discussed in Section~\ref{subsec:verse-relationships}) and three similarities (Section~\ref{subsec:ppr-instantiation}) with $\alpha{}{\in}\{0.45, 0.55, 0.65, 0.75, 0.85, 0.95\}$ for $\simppr$ and $C\in{}\{0.15, 0.25,$ $0.35,$ $0.45, 0.55, 0.65\}$  for $\simsr$.
\begin{table}[!ht]
\centering
\begin{tabularx}{\columnwidth}{XX}
\toprule
Operator & Result \\
\midrule
\textsf{Average} & \((\mathbf{a} + \mathbf{b})/2\) \\
\textsf{Concat} & \([\mathbf{a}_1, \ldots, \mathbf{a}_d, \mathbf{b}_1, \ldots, \mathbf{b}_d]\) \\
\textsf{Hadamard} & \([\mathbf{a}_1 * \mathbf{b}_1, \ldots, \mathbf{a}_d * \mathbf{b}_d]\) \\
\textsf{Weighted L1} & \([|\mathbf{a}_1 - \mathbf{b}_1|, \ldots, |\mathbf{a}_d - \mathbf{b}_d|]\) \\
\textsf{Weighted L2} & \([{(\mathbf{a}_1 - \mathbf{b}_1)}^2, \ldots, {(\mathbf{a}_d - \mathbf{b}_d)}^2] \) \\
\bottomrule
\end{tabularx}
\caption{Vector operators used for link-prediction task for each \(u,v \in V\) and corresponding embeddings \(\mathbf{a}, \mathbf{b} \in \mathbb{R}^d\).}\label{tab:lp-operators}
\vspace{-15pt}
\end{table} %
\begin{table}[!ht]
\setlength{\tabcolsep}{1pt}
\renewcommand{\aboverulesep}{0pt}
\renewcommand{\belowrulesep}{0pt}

\begin{center}
\newcolumntype{C}{>{\centering\arraybackslash}X}
\begin{tabularx}{\columnwidth}{XCCCCC}
\multicolumn{1}{l}{\emph{method}} & \textsf{Average} & \textsf{Concat} & \textsf{Hadamard} & \textsf{L1} & \textsf{L2} \\
\midrule
\small\bf\fverseemb{} & 80.06 & 79.69 & \cellcolor{cycle3!20}\underline{86.71} & 84.49 & 84.97 \\
\small\bf\verseemb{} & 79.16 & 78.79 & \cellcolor{cycle3!20}\underline{85.69} & 71.93 & 72.11 \\
\small\deepwalk{} & 68.43 & 68.06 & 66.54 & \underline{79.06} & 78.11 \\
\small\grarep{} & 74.87 & 74.91 & \underline{82.24} & 80.03 & 80.05 \\
\small\lineemb{} & 77.49 & 77.39 & \underline{77.73} & 70.55 & 71.83 \\
\small\hoppe{} & \underline{74.90} & 74.83 & 74.81 & 74.34 & 74.81 \\
\midrule
\small\bf\hypverseemb{} & 79.52 & 79.10 & \cellcolor{cycle3!20}\underline{86.15} & 76.45 & 76.72 \\ %
\small\nodetovec{} & 77.07 & 76.67 & 79.42 & \underline{81.25} & 80.85 \\
\midrule
Feature~Eng. & \multicolumn{5}{c}{77.53} \\
\bottomrule
\end{tabularx}
\end{center}
\caption{Link prediction results on the \coauthor{} coauthorship graph. Best results per method are underlined.}\label{tab:link-prediction-aco}
\vspace{-15pt}
\end{table} %
\begin{table}[!ht]
\begin{center}
\setlength{\tabcolsep}{1pt}
\renewcommand{\aboverulesep}{0pt}
\renewcommand{\belowrulesep}{0pt}
\newcolumntype{C}{>{\centering\arraybackslash}X}
\begin{tabularx}{\columnwidth}{XCCCCC}
\multicolumn{1}{l}{\emph{method}} & \textsf{Average} & \textsf{Concat} & \textsf{Hadamard} & \textsf{L1} & \textsf{L2} \\
\midrule
\small\bf\fverseemb{} & 74.94 & 74.81 & \cellcolor{cycle3!20}\underline{80.77} &  78.49 &  79.13 \\
\small\bf\verseemb{} & 73.78 & 73.66 & \cellcolor{cycle3!20}\underline{79.71} & 74.11 & 74.56 \\
\small\deepwalk{} & 70.05 & 69.92 & 69.79 & \underline{78.38} & 77.37 \\
\small\lineemb{} & \underline{75.17} &  75.13 & 72.54 & 63.77 & 64.47 \\
\small\hoppe{} & 71.89 & \underline{71.90} & 70.22 & 71.22 & 70.63 \\
\midrule
\small\bf\hypverseemb{} & 74.14 & 74.02 & \cellcolor{cycle3!20}\underline{80.26} & 73.04 & 73.53 \\
\small\nodetovec{} & 71.29 & 71.22 & 72.43 & 78.38 & \underline{78.66} \\
\midrule
Feature~Eng. & \multicolumn{5}{c}{78.84}\\
\bottomrule
\end{tabularx}
\end{center}
\caption{Link prediction results on the \dvk{} social graph. Best results per method are underlined.}\label{tab:link-prediction-vk}
\vspace{-15pt}
\end{table} %
\begin{table}[h]
\begin{center}
\setlength{\tabcolsep}{1pt}
\renewcommand{\aboverulesep}{0pt}
\renewcommand{\belowrulesep}{0pt}

\newcolumntype{C}{>{\centering\arraybackslash}X}
\begin{tabularx}{\columnwidth}{XCCCCC}
\multicolumn{1}{c}{} & \multicolumn{5}{c}{\textit{labelled nodes, \%}}\\
\cmidrule{2-6}
\multicolumn{1}{l}{\emph{method}} & 1\% & 3\% & 5\% & 7\% & 9\% \\
\midrule
\small\bf\fverseemb{} & 27.52 & \cellcolor{cycle3!20}29.83 & \cellcolor{cycle3!20}31.01 & \cellcolor{cycle3!20}31.68 & \cellcolor{cycle3!20}32.24 \\
\small\bf\verseemb{} & 27.32 & 29.42 & 30.67 & 31.32 & 31.83 \\
\small\deepwalk{} & 26.81 & 29.27 & 30.37 & 31.04 & 31.43 \\
\small\grarep{} & \cellcolor{cycle3!20}27.68 & 29.21 & 30.24 & 30.23 & 30.79 \\
\small\lineemb{} & 23.68 & 26.90 & 27.89 & 28.49 & 28.80 \\
\small\hoppe{} & 22.81 & 26.63 & 27.59 & 28.19 & 28.58 \\
\midrule
\small\bf\hypverseemb{} & 27.46 & 29.45 & 30.67 & 31.38 & 31.92 \\
\small\nodetovec{} & 27.45 & 29.66 & 30.82 & 31.54 & 32.04 \\
\bottomrule
\end{tabularx}
\end{center}
\caption{Multi-class classification results in \cocitation{} dataset.}\label{tab:classification-ak}
\vspace{-20pt}
\end{table} 
\subsection{Link Prediction}\label{subsec:link-perdiction}

Link prediction is the task of anticipating the appearance of a link between two nodes in a network. 
Conventional measures for link prediction include Adamic-Adar, Preferential attachment, Katz, and Jaccard coefficient. 
We train a Logistic regression classifier on edge-wise features obtained with the methods shown in Table~\ref{tab:lp-operators}.
For instance, for a pair of nodes \(u,v\), the \textsf{Concat} operator returns a vector as the sequential concatenation of the embeddings \(f(u)\) and \(f(v)\). 
On the \coauthor{} data, we predict new links for 2015 and 2016 co-authorships, using the network until 2014 for training; on \dvk, we predict whether a new friendship link appears between November 2016 and May 2017, using 50\% of the new links for training and 50\% for testing. We train the binary classifier by sampling non-existing edges as negative examples. Tables~\ref{tab:link-prediction-aco} and~\ref{tab:link-prediction-vk} report the attained accuracy. As a baseline, we use a logistic regression classifier trained on the respective data sets' features.

\verseemb{} with Hadamard product of vectors is consistently the best edge representation. We attribute this quality to the explicit reconstruction we achieve using noise contrastive estimation. \verseemb{} consistently outperforms the baseline in the tested datasets.
Besides, the hyperparameter-supervised \hypverseemb{} variant outruns \nodetovec{} on all datasets.

\subsection{Node Classification}\label{subsec:classification}

We now conduct an extensive evaluation on classification and report results for all the methods, where possible, with the \cocitation{}, \dvk{}, \dyoutube{}, and \dorkut{} graphs.
Node classification aims to predict of the correct node labels in a graph, as described previously in this section.

We evaluate accuracy by the Micro-F1 and Macro-F1 percentage measures. We report only Macro-F1, since we experience similar behaviors with Micro-F1.
For each dataset we conduct multiple experiments, selecting a random sample of nodes for training and leaving the remaining nodes for testing.
The results for four datasets, shown in Tables~\ref{tab:classification-ak}-\ref{tab:classification-ork}, exhibit similar trends: \verseemb{} yields predictions comparable or superior to those of the other contestants, while it scales to large networks such as~\dorkut. 
\lineemb{} outperforms \verseemb{} only in \dvk, where the gender of users is better captured using the direct neighborhood. The hyperparameter-supervised variant, \hypverseemb, is on a par with \nodetovec{} in terms of quality on \cocitation{} and \dvk; on the largest datasets \dyoutube{} and \dorkut, \hypverseemb{} keeps outperforming unsupervised alternatives, while \nodetovec{} depletes the memory. 

\begin{table}[ht!]
\begin{center}
\setlength{\tabcolsep}{1pt}
\renewcommand{\aboverulesep}{0pt}
\renewcommand{\belowrulesep}{0pt}
\newcolumntype{C}{>{\centering\arraybackslash}X}
\begin{tabularx}{\columnwidth}{XCCCCC}
\multicolumn{1}{c}{} & \multicolumn{5}{c}{\textit{labelled nodes, \%}}\\
\cmidrule{2-6}
\multicolumn{1}{l}{\emph{method}} & 1\% & 3\% & 5\% & 7\% & 9\% \\
\midrule
\small\bf\fverseemb{} & 58.32 & 61.01 & 61.74 & 62.26 & 62.50 \\
\small\bf\verseemb{} & 57.89 & 60.53 & 61.43 & 61.86 & 62.13 \\
\small\deepwalk{} & 58.22 & 60.93 & 61.79 & 62.17 & 62.49 \\
\small\lineemb{} & \cellcolor{cycle3!20}60.39 & \cellcolor{cycle3!20}62.83 & \cellcolor{cycle3!20}63.58 & \cellcolor{cycle3!20}64.01 & \cellcolor{cycle3!20}64.23 \\
\small\hoppe{} & 54.88 & 56.65 & 57.04 & 57.40 & 57.68 \\
\midrule
\small\bf\hypverseemb{} & 58.87 & 61.67 & 62.50 & 62.97 & 63.16 \\
\small\nodetovec{} & 58.85 & 61.79 & 62.62 & 63.04 & 63.30 \\
\bottomrule
\end{tabularx}
\end{center}
\caption{Multi-class classification results in \dvk{} dataset.}\label{tab:classification-vk}
\vspace{-15pt}
\end{table} %
\begin{table}[ht]
\begin{center}
\setlength{\tabcolsep}{1pt}
\renewcommand{\aboverulesep}{0pt}
\renewcommand{\belowrulesep}{0pt}
\newcolumntype{C}{>{\centering\arraybackslash}X}
\begin{tabularx}{\columnwidth}{XCCCCC}
\multicolumn{1}{c}{} & \multicolumn{5}{c}{\textit{labelled nodes, \%}}\\
\cmidrule{2-6}
\multicolumn{1}{l}{\emph{method}} & 1\% & 3\% & 5\% & 7\% & 9\% \\
\midrule
\small\bf\verseemb{} & 17.92 & 22.26 & 24.07 & 25.07 & 25.99 \\
\small\deepwalk{} & \cellcolor{cycle3!20}18.16 & 21.55 & 22.89 & 23.64 & 24.54 \\
\small\lineemb{} & 13.71 & 17.36 & 18.69 & 19.84 & 20.64 \\
\small\hoppe{} & 9.22 & 13.80 & 15.09 & 16.18 & 16.78 \\
\midrule
\small\bf\hypverseemb{} & \cellcolor{cycle3!20}18.16 & \cellcolor{cycle3!20}22.84 & \cellcolor{cycle3!20}25.40 & \cellcolor{cycle3!20}27.38 & \cellcolor{cycle3!20}29.09 \\
\bottomrule
\end{tabularx}
\end{center}
\caption{Multi-label classification results in \dyoutube{} dataset.}\label{tab:classification-yt}
\vspace{-15pt}
\end{table} %
\begin{table}[ht]
\begin{center}
\setlength{\tabcolsep}{1pt}
\renewcommand{\aboverulesep}{0pt}
\renewcommand{\belowrulesep}{0pt}
\newcolumntype{C}{>{\centering\arraybackslash}X}
\begin{tabularx}{\columnwidth}{XCCCCC}
\multicolumn{1}{c}{} & \multicolumn{5}{c}{\textit{labelled nodes, \%}}\\
\cmidrule{2-6}
\multicolumn{1}{l}{\emph{method}} & 1\% & 3\% & 5\% & 7\% & 9\% \\
\midrule
\small\bf\verseemb{} & 25.16 & 28.22 & 29.60 & 31.46 & 32.63 \\
\small\deepwalk{} & 24.21 & 27.99 & 29.63 & 30.60 & 31.27 \\
\small\lineemb{} & 26.79 & \cellcolor{cycle3!20}30.89 & 32.34 & 32.92 & 33.65 \\
\midrule
\small\bf\hypverseemb{} & \cellcolor{cycle3!20}27.73 & 30.70 & \cellcolor{cycle3!20}32.73 & \cellcolor{cycle3!20}34.00 & \cellcolor{cycle3!20}35.20 \\
\bottomrule
\end{tabularx}
\end{center}
\caption{Multi-class classification results in \dorkut{} dataset.}\label{tab:classification-ork}
\vspace{-15pt}
\end{table} 
\subsection{Node Clustering}\label{subsec:clustering}

Graph clustering detects groups of nodes with similar characteristics~\cite{newman2006,blondel2008}. We assess the embedding methods, using the \(k\)-means algorithm with \(k\)-means\texttt{++} initialization~\cite{arthur2007} to cluster the embedded points in a \(d\)-dimensional space.
Table~\ref{tab:clustering-nmi} reports the Normalized Mutual Information (NMI) with respect to the original label distribution. On \coauthor, \verseemb{} has comparable performance with \deepwalk; yet on \dvk, \verseemb{} outperforms all other methods.

We also assess graph embeddings on their ability to capture the graph community structure. 
We apply \(k\)-means with  different \(k\) values between $2$ and $50$ and select the best modularity~\cite{newman2006} score. 
Table~\ref{tab:modularity} presents our results, along with the modularity obtained by the Louvain method, the state-of-the-art modularity maximization algorithm~\cite{blondel2008}. 
\verseemb{} variants produce result almost equal that those of Louvain, outperforming previous methods, while the three methods that could manage the \dorkut{} data perform similarly.
\begin{table}[!ht]
\begin{center}
\setlength{\tabcolsep}{1pt}
\renewcommand{\aboverulesep}{0pt}
\renewcommand{\belowrulesep}{0pt}
\newcolumntype{C}{>{\centering\arraybackslash}X}
\begin{tabularx}{.65\columnwidth}{XCC}
\multicolumn{1}{l}{\emph{method}} &  \multicolumn{1}{c}{\cocitation} & \multicolumn{1}{c}{\dvk} \\
\midrule
\bf\fverseemb{} & 33.22 & \cellcolor{cycle3!20}9.24 \\
\bf\verseemb{} & 32.93 & 7.62 \\
\deepwalk{} & \cellcolor{cycle3!20}34.33 & 7.59 \\
\lineemb{} & 18.79 & 7.49 \\
\grarep{} & 27.43 & --- \\
\hoppe{} & 19.05 & 6.47 \\
\midrule
\bf\hypverseemb{} & 33.24 & 8.77 \\
\nodetovec{} & 32.84 & 8.05 \\
\midrule
Louvain & 30.73 & 4.54 \\
\bottomrule
\end{tabularx}
\end{center}
\caption{Node clustering results in terms of NMI.}\label{tab:clustering-nmi}
\vspace{-14pt}
\end{table} %
\begin{table}[!ht]
\begin{center}
\setlength{\tabcolsep}{1pt}
\renewcommand{\aboverulesep}{0pt}
\renewcommand{\belowrulesep}{0pt}
\newcolumntype{C}{>{\centering\arraybackslash}X}
\begin{tabularx}{\columnwidth}{XCCCCC}
\multicolumn{1}{l}{\emph{method}} & \cocitation{} & \coauthor{} & \dvk{} & \dyoutube{} & \dorkut{} \\
\midrule
\small\bf\fverseemb{} & \cellcolor{cycle3!20}70.12 & \cellcolor{cycle3!20}80.95 & 44.59 & --- & --- \\
\small\bf\verseemb{} & 69.43 & 79.25 & 45.78 & 67.63 & 42.64 \\
\small\deepwalk{} & 70.04 & 73.83 & 43.30 & 58.08 & \cellcolor{cycle3!20}44.66 \\
\small\lineemb{} & 60.02 & 71.58 & 39.65 & 63.40 & 42.59 \\
\small\grarep{} & 67.61 & 77.40 & --- & --- & --- \\
\small\hoppe{} & 42.45 & 69.57 & 21.70 & 37.94 & --- \\
\midrule
\small\bf\hypverseemb{} & 69.81 & 79.31 & \cellcolor{cycle3!20}45.84 & \cellcolor{cycle3!20}69.13 & --- \\
\small\nodetovec{} & 70.06 & 75.78 & 44.27 & --- & --- \\
\midrule
Louvain & 72.05 & 84.29 & 46.60 & 71.06 & --- \\
\bottomrule
\end{tabularx}
\end{center}
\caption{Node clustering results in terms of modularity.}\label{tab:modularity}
\vspace{-14pt}
\end{table} 
\subsection{Graph Reconstruction}\label{subsec:graph-resonstruction}

Good graph embeddings should preserve the graph structure in the embedding space. We evaluate the performance of our method on reconstructing the graph's adjacency matrix. Since each adjacent node should be close in the embedding space, we first sort any node other than the one considered by decreasing cosine distance among the vectors. Afterwards, we take a number of nodes equal to the actual degree of the node in the graph and connect to the considered node to create the graph structure.

Table~\ref{tab:graphrec} reports the relative accuracy measured as the number of correct nodes in the neighborhood of a node in the embedding space.
Again, \verseemb{} performs comparably well; its exhaustive variant, \fverseemb, which harnesses the full similarity does even better; however, the top performer is \hypverseemb{}, which achieves the obtained result when instantiated to the Adjacency Similarity. This result is unsurprising, given that the adjacency similarity measure tailors \hypverseemb{} for the task of graph reconstruction. 
\begin{table}[ht]
\begin{center}
\setlength{\tabcolsep}{1pt}
\renewcommand{\aboverulesep}{0pt}
\renewcommand{\belowrulesep}{0pt}
\newcolumntype{C}{>{\centering\arraybackslash}X}
\begin{tabularx}{\columnwidth}{XCCCCC}
\multicolumn{1}{l}{\emph{method}} & \cocitation{} & \coauthor{} & \dvk{} & \dyoutube{} & \dorkut{} \\
\midrule
\small\bf\fverseemb{} & 88.96 & 98.20 & 66.45 & --- & --- \\
\small\bf\verseemb{} & 58.73 & 74.30 & 50.18 & 28.64 & 18.39 \\
\small\deepwalk{} & 51.54 & 68.44 & 43.04 & 32.21 & 19.75 \\
\small\lineemb{} & 23.32 & 62.01 & 42.80 & 17.76 & 10.82 \\
\small\grarep{} & 67.61 & 77.40 & --- & --- & --- \\
\small\hoppe{} & 25.88 & 49.70 & 12.01 & 33.42 & --- \\
\midrule
\small\bf\hypverseemb{} & \cellcolor{cycle3!20}97.53 & \cellcolor{cycle3!20}98.91 & \cellcolor{cycle3!20}78.38 & \cellcolor{cycle3!20}38.34 & \cellcolor{cycle3!20}28.81 \\ %
\small\nodetovec{} & 66.35 & 72.70 & 53.70 & --- & --- \\
\bottomrule
\end{tabularx}
\end{center}
\caption{Graph reconstruction \% for all datasets.}\label{tab:graphrec}
\vspace{-14pt}
\end{table} 

\subsection{Parameter Sensitivity}\label{subsec:parameter-study}

We also evaluate the sensitivity of \verseemb{} to parameter choice.
{Figures~\ref{sfig:dimensionality},\ref{sfig:alpha} depict node classification performance in terms of Micro-F1 on {the} \dblogcatalog dataset, with 10\% of nodes labeled.}

The dimensionality \(d\) determines the size of the embedding, and hence the possibility to compute more fine-grained representations. 
The performance grows linearly as the number of dimensions approaches \(128\), while with larger \(d\) we observe no further improvement. 
Sampled \verseemb{} instead, performs comparably better than \fverseemb{} in low dimensional spaces, but degrades as \(d\) becomes larger than \(128\); this behavior reflects a characteristic of node sampling that tends to preserve similarities of close neighborhoods in low-dimensional embeddings, while the \verseemb{} leverages the entire graph structure for larger dimensionality

The last parameter we study is the damping factor \(\alpha{}\) which amounts to the inverse of the probability of restarting random walks from the initial node.
As shown in Figure~\ref{sfig:alpha}, the quality of {classification accurary} is quite robust with respect to \(\alpha{}\) {for both \verseemb and \fverseemb}, only compromised by extreme values. An \(\alpha{}\) value close to \(0\) reduces PPR to an exploration of the immediate neighborhood of the node. 
On the other hand, a value close to \(1\) amounts to regular PageRank, deeming all nodes as equally important. 
This result vindicates our work and distinguishes it from previous methods based on local neighborhood expansion.

\subsection{Scalability}\label{subsec:scalability}

We now present runtime results on synthetic graphs of growing size, generated by the Watts Strogatz model~\cite{watts1998}, setting \verseemb{} against scalable methods with C\texttt{++} implementations, namely \deepwalk{}, \lineemb, and \nodetovec{}. For each method, we report the total wall-clock time, with graph loading and necessary preprocessing steps included. We used \lineemb{}-2 time for fair comparison. As Figure~\ref{fig:new-multiplot} shows, \verseemb{} is comfortably the most efficient and scalable method, processing \(10^6\) nodes in about 3 hours, while \deepwalk{} and \lineemb{} take from 6 to 15 hours.
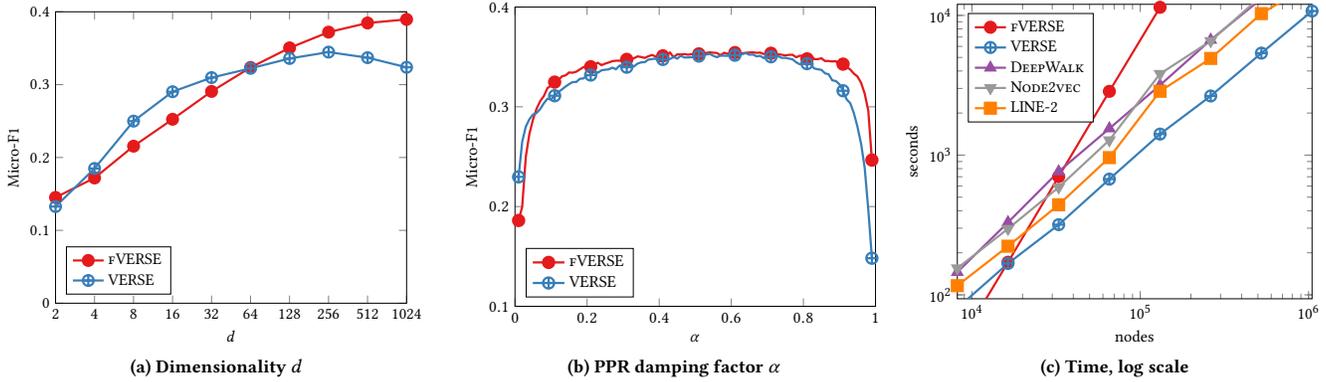
\begin{figure*}[ht!]
\centering
\subfloat[Dimensionality \(d\)]{
\resizebox{0.32\textwidth}{!}{%
\begin{tikzpicture}
        \begin{axis}[
    	xmode=log,
        log basis x={2},
        xlabel=\(d\),
        ylabel=$\text{Micro-F1}$,
        xmin=2,
        xmax=1024,
        xticklabel={\pgfmathparse{2^\tick}\pgfmathprintnumber[fixed]{\pgfmathresult}},
        x tick label style={/pgf/number format/.cd,%
          scaled x ticks = false,
          set thousands separator={},
          fixed
        },
        ymin=0,
        ymax=0.4,
        legend cell align=left,
        legend pos=south west
]
\addplot[very thick,color=cycle1,mark=*,mark size=3pt] coordinates {
(2, 0.14518429008456057)
(4, 0.17191909323160551)
(8, 0.21554629543396323)
(16, 0.25246253004659508)
(32, 0.29070114490914462)
(64, 0.32359304463467348)
(128, 0.3504865787429266)
(256, 0.37207863786289935)
(512, 0.38462954962445139)
(1024, 0.38953561303366208)
};\addlegendentry{\fverseemb};\label{line:v-dim};
\addplot[very thick,color=cycle2,mark=oplus,mark size=3pt] coordinates {
(2, 0.13263410215834281)
(4, 0.18496560458991274)
(8, 0.24990881632366194)
(16, 0.29027811015720706)
(32, 0.30967144140984226)
(64, 0.32232786256279)
(128, 0.33591224774893513)
(256, 0.34466234919945421)
(512, 0.33701671959642732)
(1024, 0.32380749276796034)
};\addlegendentry{\verseemb};\label{line:sv-ppr-dim};
\end{axis}
\end{tikzpicture} %
}\label{sfig:dimensionality}
}\hfill
\subfloat[PPR damping factor \(\alpha{}\)]{
\resizebox{0.32\textwidth}{!}{%
\begin{tikzpicture}
        \begin{axis}[
        ylabel=$\text{Micro-F1}$,
        xmin=0,
        xmax=1,
        ymin=0.1,
        ymax=0.4,
        xlabel=\(\alpha{}\),
        mark repeat={10},
        legend cell align=left,
        legend pos=south west
]
\addplot[very thick,color=cycle1,mark=*,mark size=3pt] coordinates {
(0.01, 0.18604260638414297)
(0.02, 0.19972578940035046)
(0.03, 0.24988303243188187)
(0.04, 0.27298697899163993)
(0.05, 0.28726728978776933)
(0.06, 0.29823985483199394)
(0.07, 0.30574619205296472)
(0.08, 0.31076663664000198)
(0.09, 0.31873375662467174)
(0.1, 0.32117069091136197)
(0.11, 0.32475122351738311)
(0.12, 0.32634136215976489)
(0.13, 0.33053645608498444)
(0.14, 0.33147293978733133)
(0.15, 0.3344175096605187)
(0.16, 0.33417555208477734)
(0.17, 0.33522560559055092)
(0.18, 0.33752584699337973)
(0.19, 0.33977179524790579)
(0.2, 0.3399309030294298)
(0.21, 0.34016656196544331)
(0.22, 0.34268860584080874)
(0.23, 0.3441429069592899)
(0.24, 0.34175145903745358)
(0.25, 0.34326140428620983)
(0.26, 0.3444175849638359)
(0.27, 0.34573948300268947)
(0.28, 0.34619328174753805)
(0.29, 0.34626617311209779)
(0.3, 0.34567174359103853)
(0.31, 0.34752512180277728)
(0.32, 0.34729238266989154)
(0.33, 0.34870231610324692)
(0.34, 0.35054912968435076)
(0.35, 0.34990873827976648)
(0.36, 0.34900380258711439)
(0.37, 0.35067771204750742)
(0.38, 0.35200964432820359)
(0.39, 0.35042259071310761)
(0.4, 0.35040962418354549)
(0.41, 0.35108700407747812)
(0.42, 0.35262584464630747)
(0.43, 0.3543718747722473)
(0.44, 0.34970932398293364)
(0.45, 0.35255310797729872)
(0.46, 0.35251086183275082)
(0.47, 0.35286983303858399)
(0.48, 0.35220844649689587)
(0.49, 0.35282441222360406)
(0.5, 0.35226307092547637)
(0.51, 0.35303718839964843)
(0.52, 0.35425686027736175)
(0.53, 0.35393664953626197)
(0.54, 0.35463282050691908)
(0.55, 0.35421387322773434)
(0.56, 0.35398078230320024)
(0.57, 0.35167754584993899)
(0.58, 0.35416100823310276)
(0.59, 0.35320439732096465)
(0.6, 0.3534994965064624)
(0.61, 0.35434098485782073)
(0.62, 0.35326096222014086)
(0.63, 0.35358087921309889)
(0.64, 0.35412493153560121)
(0.65, 0.35441586118642687)
(0.66, 0.35461898532368785)
(0.67, 0.35410974619584751)
(0.68, 0.35437053817198505)
(0.69, 0.35290734362059178)
(0.7, 0.35327018663549464)
(0.71, 0.35369760739378925)
(0.72, 0.35215476128333023)
(0.73, 0.35262958525078053)
(0.74, 0.35220662304668854)
(0.75, 0.35148935042995855)
(0.76, 0.35212164604139895)
(0.77, 0.35215166344954213)
(0.78, 0.34966496544810077)
(0.79, 0.35120038869760095)
(0.8, 0.35105477905693755)
(0.81, 0.3480531838155041)
(0.82, 0.35218619957228503)
(0.83, 0.34773886060328379)
(0.84, 0.34778302894173641)
(0.85, 0.348661878268757)
(0.86, 0.34589762873673352)
(0.87, 0.34595258577401183)
(0.88, 0.34547025521817076)
(0.89, 0.34576653901961607)
(0.9, 0.34200314083144839)
(0.91, 0.34280073872513528)
(0.92, 0.33958556281384367)
(0.93, 0.33740975642763493)
(0.94, 0.33322325742159692)
(0.95, 0.33039916805806402)
(0.96, 0.32406060515094126)
(0.97, 0.31326505804237992)
(0.98, 0.29617761622961303)
(0.99, 0.24663134830856806)
(0.99, 0.24663134830856806)
(0.99, 0.24663134830856806) %
};\addlegendentry{\fverseemb};\label{line:v-alpha};
\addplot[very thick,color=cycle2,mark=oplus,mark size=3pt] coordinates {
(0.01, 0.22975753416311967)
(0.02, 0.26503412959430267)
(0.03, 0.28019730734414566)
(0.04, 0.2875679382167895)
(0.05, 0.29228999705295794)
(0.06, 0.29385147959590752)
(0.07, 0.29782613226063037)
(0.08, 0.30348763195820311)
(0.09, 0.30808659742091954)
(0.1, 0.30734196605390629)
(0.11, 0.31110071401804668)
(0.12, 0.31475587177022474)
(0.13, 0.31737260658833516)
(0.14, 0.3202566366737557)
(0.15, 0.32275057430417314)
(0.16, 0.32477796974571632)
(0.17, 0.32491997099313363)
(0.18, 0.32807513364476998)
(0.19, 0.32845234516027005)
(0.2, 0.33147031393052018)
(0.21, 0.33187212407581645)
(0.22, 0.33569312454455225)
(0.23, 0.33417721877416545)
(0.24, 0.3366026522807809)
(0.25, 0.33715496774141529)
(0.26, 0.33752877304758677)
(0.27, 0.33926831844664834)
(0.28, 0.3376420729773657)
(0.29, 0.34134447643005977)
(0.3, 0.3408112277760823)
(0.31, 0.33970916271138257)
(0.32, 0.34251128019740001)
(0.33, 0.34044137083820369)
(0.34, 0.34331665262630046)
(0.35, 0.34453792815198186)
(0.36, 0.34481249339063247)
(0.37, 0.34733361862361084)
(0.38, 0.34783099340166468)
(0.39, 0.3483678816009621)
(0.4, 0.34523689348286013)
(0.41, 0.34747389176439414)
(0.42, 0.34776070116868985)
(0.43, 0.34796726613984219)
(0.44, 0.34994441232686113)
(0.45, 0.3483230536777438)
(0.46, 0.35118117384712733)
(0.47, 0.34942905089923831)
(0.48, 0.34948619477438103)
(0.49, 0.34970360670235451)
(0.5, 0.35037195526530995)
(0.51, 0.3508787551302584)
(0.52, 0.35287212055072664)
(0.53, 0.35232063916232798)
(0.54, 0.35309975791045117)
(0.55, 0.35113302662062551)
(0.56, 0.34972211163366479)
(0.57, 0.35127679209965312)
(0.58, 0.35120181213853235)
(0.59, 0.35112110343639258)
(0.6, 0.35231079530739123)
(0.61, 0.35173120402148661)
(0.62, 0.3515533447568675)
(0.63, 0.35268720849616464)
(0.64, 0.35385105510897796)
(0.65, 0.35219908697060293)
(0.66, 0.35243928531791491)
(0.67, 0.35060231676381298)
(0.68, 0.35111661129274097)
(0.69, 0.34880832290810188)
(0.7, 0.35136442687713598)
(0.71, 0.35017236713258504)
(0.72, 0.34976663205746539)
(0.73, 0.35077298770527249)
(0.74, 0.34902706124507943)
(0.75, 0.34908901338990972)
(0.76, 0.34834952702156285)
(0.77, 0.34759519619196699)
(0.78, 0.34462377905297364)
(0.79, 0.34262588750897371)
(0.8, 0.3439383265539564)
(0.81, 0.34327620222891603)
(0.82, 0.34280646281830163)
(0.83, 0.33843113406940412)
(0.84, 0.33766130601094507)
(0.85, 0.33707484414379252)
(0.86, 0.33358273494299895)
(0.87, 0.33261960577893857)
(0.88, 0.32663587305756087)
(0.89, 0.32581617172179611)
(0.9, 0.32262477850486249)
(0.91, 0.31609083574492974)
(0.92, 0.30742105382658075)
(0.93, 0.3035879518477233)
(0.94, 0.29359242923086437)
(0.95, 0.27997803451114156)
(0.96, 0.26306368002809388)
(0.97, 0.23890619097415761)
(0.98, 0.19700483554747894)
(0.99, 0.14823134456706669)
(0.99, 0.14823134456706669)
(0.99, 0.14823134456706669) %
};\addlegendentry{\verseemb};\label{line:sv-alpha};
\end{axis}
\end{tikzpicture} %
}\label{sfig:alpha}
}
\subfloat[Time, log scale]{
\resizebox{0.32\textwidth}{!}{%
\begin{tikzpicture}
        \begin{axis}[
    	xmode=log,
    	ymode=log,
        log basis x={10},%
        log basis y={10},
        xlabel=nodes,
        ylabel=$\text{seconds}$,
        xmax=1048576,
        xmin=8192,
        ymin=93.75,
        ymax=12000, %
        legend cell align=left,
        legend pos=north west
]
\addplot[very thick,color=cycle1,mark=*,mark size=3pt] coordinates {
(128, 0.010893451)
(256, 0.043573802)
(512, 0.174295209)
(1024, 0.5402453)
(2048, 1.93025458)
(4096, 9.914050839)
(8192, 42.556642763)
(16384, 171.190)
(32768, 701.639)
(65536, 2859.214)
(131072, 11422.61079)
};\addlegendentry{\fverseemb};\label{line:verse-time}
\addplot[very thick,color=cycle2,mark=oplus,mark size=3pt] coordinates {
(128, 2.339124599)
(256, 3.507897897)
(512, 6.143420775)
(1024, 11.32098132)
(2048, 22.976270655)
(4096, 41.674204901)
(8192, 80.417648415)
(16384, 167.711005514)
(32768, 317.947839123)
(65536, 673.501646848)
(131072, 1414.056190692)
(262144, 2650.066212017)
(524288, 5368.64508846)
(1048576, 10739.29017693)
};\addlegendentry{\verseemb};\label{line:sverse-time}
\addplot[very thick,color=cycle4,mark=triangle*,mark size=3pt] coordinates {
(128, 3.084920975)
(256, 5.864827386)
(512, 11.263897347)
(1024, 22.905526018)
(2048, 41.598715712)
(4096, 83.331275662)
(8192, 145.416250483)
(16384, 330.264641086)
(32768, 763.432211581)
(65536, 1549.492125138)
(131072, 3167.453462369)
(262144, 6674.6724824153)
(524288, 13152.20001096)
};\addlegendentry{\deepwalk};\label{line:dw-time}
\addplot[very thick,color=cycle6,mark=triangle*,mark options={rotate=180},mark size=3pt] coordinates {
(128, 2.46780945)
(256, 6.264020225)
(512, 14.41531385)
(1024, 21.17842084)
(2048, 40.83916522)
(4096, 76.91844478)
(8192, 155.133946)
(16384, 297.2511844)
(32768, 588.4850422)
(65536, 1279.575271)
(131072, 3818.578287)
(262144, 6558.542229)
(524288, 13723.10139)
};\addlegendentry{\nodetovec};\label{line:ntv-time}
\addplot[very thick,color=cycle5,mark=square*,mark options={rotate=90},mark size=3pt] coordinates {
(128, 1.48068567)
(256, 3.758412135)
(512, 8.64918831)
(1024, 15.88381563)
(2048, 30.62937392)
(4096, 57.68883359)
(8192, 116.3504595)
(16384, 222.9383883)
(32768, 441.3637817)
(65536, 959.6814534)
(131072, 2863.933715)
(262144, 4918.906672)
(524288, 10292.32604)
(1048576, 17029.5)
};\addlegendentry{\lineemb-2};\label{line:line-time}
\end{axis};
\end{tikzpicture} %
}\label{sfig:scalability-loglog}
}
\caption{{Classification performance} of various parameters in Fig.~\ref{sfig:dimensionality},~\ref{sfig:alpha} and scalability of different methods in Fig.~\ref{sfig:scalability-loglog}.}%
\label{fig:new-multiplot}
\end{figure*} %

\subsection{Visualization}\label{subsec:vizualization}

Last, we show how different embeddings are visualized on a plane.
We apply t-SNE~\cite{maaten2009} with default parameters to each embedding for {{a} subset} of 1500 nodes from the \cocitation{} dataset, equally distributed in 5 classes (i.e., conferences); we set the density areas for each class by Kernel Density Estimation. {Figure~\ref{fig:visualization} {depicts} the result.} VERSE produces well separated clusters with low noise, even finding distinctions among papers of the same community, namely ICDE (\tikz\draw[cycle5,fill=cycle5] (0,0) circle (.7ex);) and VLDB (\tikz\draw[cycle2,fill=cycle2] (0,0) circle (.7ex);).

\begin{figure}
\centering
\resizebox{0.9\columnwidth}{!}{\includegraphics{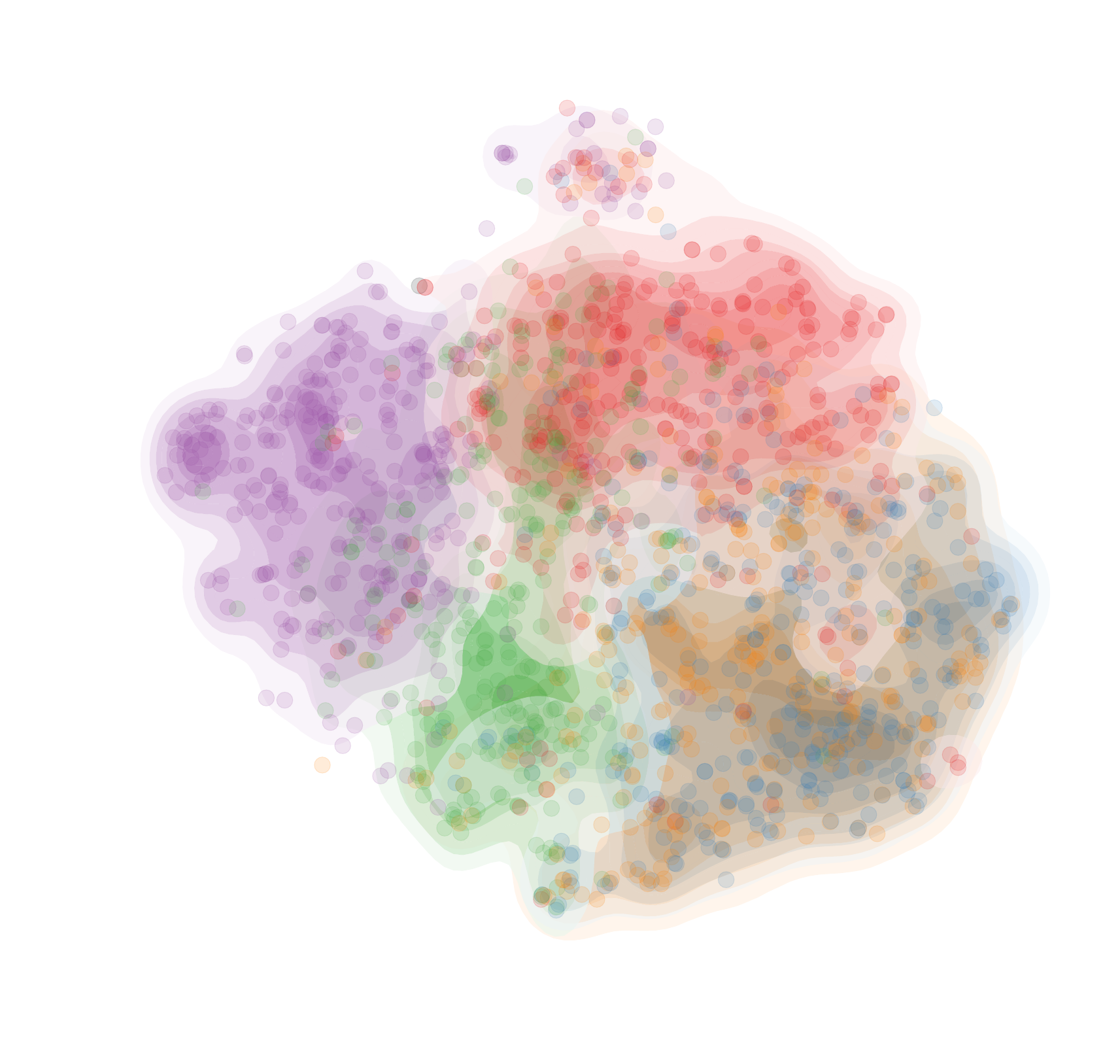}}\label{sfig:ak-v-ppr}
\caption{Visualizations of a subset of nodes from \cocitation{} graph with selected conferences: \protect\tikz\protect\draw[cycle2,fill=cycle2] (0,0) circle (.7ex);~VLDB, \protect\tikz\protect\draw[cycle5,fill=cycle5] (0,0) circle (.7ex);~ICDE, \protect\tikz\protect\draw[cycle3,fill=cycle3] (0,0) circle (.7ex);~KDD, \protect\tikz\protect\draw[cycle1,fill=cycle1] (0,0) circle (.7ex);~WWW, and \protect\tikz\protect\draw[cycle4,fill=cycle4] (0,0) circle (.7ex);~NIPS\@. Note that the number of nodes per class is the same for all conferences.}\label{fig:visualization}
\end{figure}  %

\balance %
\section{Conclusions}\label{sec:conclusion}

{We introduced a new perspective on graph embeddings: to be expressive, a graph embedding should capture \emph{some} similarity measure among nodes. Armed with this perspective, we developed a scalable embedding algorithm, \verseemb{}.}
In a departure from previous works in the area, \verseemb{} aims to reconstruct the distribution of any chosen similarity measure for each graph node.
Thereby, \verseemb{} brings in its scope a global view of the graph, while substantially reducing the number of parameters required for training.
\verseemb{} attains linear time complexity, {hence it scales} to large real graphs, while it only requires space to store the graph.
{Besides, we have shed light on some previous works on graph embeddings, looking at them and interpreting them through the prism of vertex similarity.}

Our thorough experimental study shows that,
{even instantiated with PPR as a {\em default\/} similarity notion,} \verseemb{} consistently outperforms state\hyp{}of\hyp{}the\hyp{}art approaches for graph embeddings in a plethora of graph tasks{, while a hyperparameter\hyp{}supervised variant does even better.}
{Thus, we have provided strong evidence that embeddings genuinely based on vertex similarity address graph mining challenges better than others.}

\newpage
\balance
\bibliographystyle{ACM-Reference-Format}
\bibliography{main}
\end{document}